\documentclass[12pt]{article}
\usepackage{amsmath, amsthm, amssymb}
\usepackage{bbm}
\usepackage{fullpage}
\usepackage{xcolor}
\usepackage{tikz}
\usepackage[hidelinks,final=true,colorlinks=false]{hyperref}
\usepackage[capitalise]{cleveref}
\usepackage{algorithm}
\usepackage{algpseudocode}
\usepackage{caption}
\usepackage{subcaption}
\usepackage{thm-restate}
\usepackage{mathtools}

\renewcommand{\Pr}{\mathbb{P}}
\newcommand{\R}{\mathbb{R}}

\newcommand{\E}{\mathbb{E}}
\newcommand{\cE}{\mathcal{E}}

\newcommand{\ket}[1]{| #1 \rangle}
\newcommand{\bra}[1]{\langle #1|}
\newcommand{\braket}[1]{\langle #1 \rangle}

\DeclareMathOperator{\poly}{poly}

\DeclareMathOperator{\linspan}{span}
\DeclareMathOperator{\ind}{\mathbbm{1}}
\DeclareMathOperator{\swap}{\mathrm{SWAP}}

\newtheorem{theorem}{Theorem}
\newtheorem{lemma}{Lemma}
\newtheorem{corollary}{Corollary}
\newtheorem{conjecture}{Conjecture}

\newcommand{\norm}[1]{\left\lVert#1\right\rVert}
\newcommand{\ones}{\mathbf{1}}
\newcommand{\HT}{\mathrm{HT}}
\newcommand{\CT}{\mathrm{CT}}
\newcommand{\ET}{\mathrm{ET}}
\newcommand{\EHT}{\mathrm{EHT}}
\newcommand{\esc}{\mathrm{esc}}
\newcommand{\eps}{\varepsilon}
\newcommand{\tO}{\widetilde{O}}
\newcommand{\cl}{\mathrm{cl}}

\title{Elfs, trees and quantum walks}

\author{Simon Apers\thanks{Universit\'e Paris Cit\'e, CNRS, IRIF, Paris, France. \texttt{apers@irif.fr}}
\and Stephen Piddock\thanks{Department of Computer Science, Royal Holloway, University of London, United Kingdom. \texttt{stephen.piddock@rhul.ac.uk}}}

\date{\vspace{-10mm}}

\begin{document}

\begin{titlepage}
	\maketitle \pagenumbering{roman}
	
\begin{abstract}
We study an elementary Markov process on graphs based on electric flow sampling (elfs).
The elfs process repeatedly samples from an electric flow on a graph.
While the sinks of the flow are fixed, the source is updated using the electric flow sample, and the process ends when it hits a sink vertex.

We argue that this process naturally connects to many key quantities of interest.
E.g., we describe a random walk coupling which implies that the elfs process has the same arrival distribution as a random walk.
We also analyze the electric hitting time, which is the expected time before the process hits a sink vertex.
As our main technical contribution, we show that the electric hitting time on trees is logarithmic in the graph size and weights.

The initial motivation behind the elfs process is that quantum walks can sample from electric flows, and they can hence implement this process very naturally.
This yields a quantum walk algorithm for sampling from the random walk arrival distribution, which has widespread applications.
It complements the existing line of quantum walk search algorithms which only return an element from the sink, but yield no insight in the distribution of the returned element.
By our bound on the electric hitting time on trees, the quantum walk algorithm on trees requires quadratically fewer steps than the random walk hitting time, up to polylog factors.
\end{abstract}

 	\setcounter{tocdepth}{2}
 	\newpage
    \renewcommand{\baselinestretch}{0.9}\normalsize
    \tableofcontents
    \renewcommand{\baselinestretch}{1.0}\normalsize
	\newpage
\end{titlepage}

\newpage
\pagenumbering{arabic}

\section{Introduction and summary}

It has long been known that there is an elegant and fundamental connection between electric networks and random walks on graph, justified through the theory of harmonic functions~\cite{doyle1984random}.
This connection gives precise characterizations of random walk and Markov chain properties in terms of the associated electric network properties \cite{chandra1996electrical}, and it has been instrumental in the study of recurrence and percolation properties \cite{lyons2017probability}, and the design of randomized and dynamic algorithms \cite{broder1989generating,gao2022fully}.
In this work we consider a new process called the \emph{electric flow sampling process} or \emph{elfs process}, which repeatedly samples from an electric flow on a graph.
While the sinks of the flow are fixed, the source is updated using the electric flow sample, and the process ends when a sample hits a sink vertex.
We analyze this process, uncovering many connections with the usual random walks and electric networks framework.
One highlight is a coupling of the elfs process to the random walk from the initial source to the sinks.
This shows that the process has the same arrival distribution as the random walk.
The coupling also reveals a close connection to a new quantity called the random walk ``escape time''.
We conjecture, and provide evidence, that a random walk can sample from an electric flow in a number of steps given by the escape time.
Then we analyze the ``electric hitting time'', which is the expected number of electric flow samples before the process hits a sink vertex (put more mundanely, how long until the elfs hit the sink?).
As the main technical contribution, we show that on trees the electric hitting time is logarithmic in the graph size and edge weights, and we reveal a connection to the entropy of the arrival distribution.

The elfs process is alternatively (and originally) motivated by a more recent but similarly profound connection between electric networks and \emph{quantum} walks.
Quantum walks are the quantum analogue of random walks on graphs, and they are instrumental in the design of quantum algorithms and the study of quantum computing and quantum physics more broadly.
The connection between quantum walks and electric networks seems to first have been described by Belovs \cite{belovs2013quantum,belovs2013time}, who used it in a quantum walk search algorithm for 3-distinctness.
It was later used and refined with applications to backtracking algorithms \cite{montanaro18quantum,jarret2018improved}, the traveling-salesman problem \cite{moylett2017quantum}, 2-player games \cite{ambainis2017quantum} and branch-and-bound algorithms \cite{montanaro2020quantum}.
The direct impulse for this work came from work by Piddock \cite{piddock2019quantum}, in which it was shown that quantum walks can be used very naturally to sample from electric flows.
We show that this requires a number of quantum walk steps bounded by the \emph{square root} of the random walk escape time, providing a quadratic speedup over the conjectured random walk complexity.
By repeatedly using a quantum walk for electric flow sampling, we obtain a quantum walk algorithm that approximately samples from the random walk arrival distribution.
This opens up new avenues for finding quantum speedups over random walk algorithms.
Using our bound on the electric hitting time on trees, we show that the resulting quantum walk algorithm samples from the arrival distribution on trees in quadratically fewer steps than the random walk hitting time, up to polylogarithmic factors.

\subsection{Elfs}

Here we describe in more detail our results concerning the elfs process and its connection to random walks.
Consider a graph $G = (V,E,w)$ with nonnegative weights $w_e \geq 0$ for $e \in E$.
The unit electric flow between a source vertex $s$ and sink $M \subseteq V$ is defined as the unique unit flow $f:E \to \R$ from $s$ to $M$ that minimizes the energy
\[
\cE(f)
= \sum_e \frac{f_e^2}{w_e}.
\]
The resulting energy of the electric flow is called the effective resistance $R_s = \cE(f)$.
If we interpret the graph as an electric network, by replacing every edge $e$ by a resistor with conductance $w_e$, then the electric flow is the flow resulting from applying a voltage of $R_s$ between $s$ and $M$.
Alternatively, it is the unique unit potential flow from $s$ to $M$ induced by potentials or \emph{voltages} $v_x$, in the sense that $f_{(x,y)} = (v_x-v_y) w_{x,y}$.

The elfs process is now defined as follows:

\begin{algorithm}
\caption*{\bf Electric flow sampling (elfs) process} \label{alg:elfs}
\begin{algorithmic}
\vspace{1mm}
\State
From an initial source vertex $s$ and fixed sink $M \subseteq V$, repeat the following:
\begin{enumerate}
\item
Let $f$ denote the unit electric flow from source $s$ to sink $M$.\newline
Sample an edge $e$ with probability proportional to $f_e^2/w_e$.
\item
Pick a random endpoint $x$ from $e$ and set the source $s = x$.
\end{enumerate}
The process ends when step 2.~picks a vertex $x \in M$.
\end{algorithmic}
\end{algorithm}

We define the limit or \emph{arrival distribution} $\mu$ as the distribution over $M$ when the process ends.
The \emph{electric hitting time} $\EHT_s$ from $s$ is the expected number of iterations or electric flow samples it takes from $s$ to hit $M$.

\paragraph{Arrival distribution and electric hitting time}

With some effort, we can find a closed and algebraic form for the transition matrix of the elfs process, expressed in terms of the graph Laplacian.
This leads to some (a priori) mysterious conclusions.

Consider the elfs process with initial source vertex $s$ and sink $M$, and let $f$ denote the unit electric flow from $s$ to $M$.
From the algebraic expression, we find that the arrival distribution equals
\[
\mu_x
= \sum_y f_{y,x}.
\]
I.e., the process ends at vertex $x \in M$ with probability equal to the electric flow from $s$ into vertex $x$.
From the connection between random walks and electric networks, we know that this equals the arrival distribution of a \emph{random walk} from $s$ to $M$.

The algebraic expression also gives some insight into the electric hitting time.
Let $\HT_s$ denote the random walk hitting time from $s$ to $M$, let $v_x$ denote the electric potential at $x$ in the unit electric $s$-$M$ flow, and consider a new quantity defined as
\[
\ET_s
= \frac{1}{R_s} \sum_x v^2_x d_x.
\]
We call this quantity the \emph{escape time} from $s$, for reasons to be clarified later.
The expression should be compared to a similar expression for the hitting time, which is $\HT_s = \sum_x v_x d_x$.
Now if $Y_1$ is the random endpoint picked after sampling an edge from the electric $s$-$M$ flow, then we show that
\begin{equation} \label{eq:intro-ET}
\E[\HT_{Y_1}]
= \HT_s - \frac{1}{2} \ET_s,
\end{equation}
where the expectation is over the random vertex $Y_1$.
So, in expectation, a single electric flow sample knocks off $\ET_s/2$ from the hitting time.
Now let $\{Y_0 = s, Y_1, \dots, Y_\rho \in M \}$ denote the Markov chain corresponding to the elfs process (i.e., the random sequence of sources), so that $\E[\rho] = \EHT_s$ equals the electric hitting time.
Then a direct consequence is that
\begin{equation} \label{eq:sum-ET}
\E\left[ \sum_{t=0}^{\rho-1} \ET_{Y_t} \right]
= 2 \HT_s.
\end{equation}
Using that $\ET_x \geq 1$ for $x \notin M$,\footnote{This follows from the fact that $\ET_s \geq \frac{1}{R_s} v^2_s d_s = R_s d_s \geq 1$ for $s \notin M$.} this implies for instance that the electric hitting time $\EHT_s \leq 2 \HT_s$, and this inequality is tight for a graph consisting of a single edge.

\paragraph{Random walk coupling}

We obtain a more intelligible explanation for these observations (and terminologies) by proving that there exists a coupling between a random walk from $s$ to $M$ and the elfs process.
Consider the respective Markov chains
\[
\{ X_0 = s,X_1,\dots,X_\tau \in M \}
\quad \text{and} \quad \{ Y_0 = s,Y_1,\dots,Y_\rho \in M \}.
\]
We show that there exist random stopping times $0 < \nu_1 < \nu_2 < \dots < \nu_\rho = \tau$ such that
\[
\left( \{X_0=s, X_1,\dots,X_\tau \in M\},\,
\{Y_0=s,Y_1 = X_{\nu_1},\dots,Y_\rho = X_{\nu_\rho} \in M\} \right)
\]
describes a coupling between the two processes.
In other words, we can think of the elfs process as a subsampling of the random walk from $s$ to $M$.
A direct consequence of this coupling is that both processes have the same arrival distribution (i.e., $\Pr[Y_\rho = y] = \Pr[X_\tau = y]$ for all $y$), explaining our earlier observation.

\begin{figure}[htb]
\centering
\includegraphics[width=.6\textwidth]{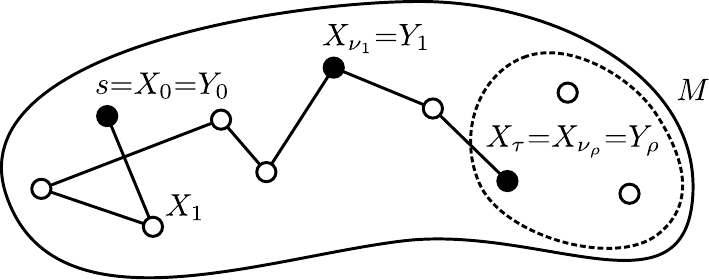}
\caption{Coupling between a random walk and elfs process through stopping rules $\nu_1<\dots<\nu_\rho$.}
\label{fig:intro-coupling}
\end{figure}

In the terminology of \cite{lovasz1995efficient}, the stopping time $\nu_1$ is a ``local stopping rule'' from $s$ to the elfs distribution.
The rule has a ``halting state'' (any sink vertex), and hence it is ``mean-optimal'': its expected length $\E[\nu_1]$ is optimal among all stopping rules, and it equals the ``access time'' from $s$ to the elfs distribution.

We show that this access time is equal to half the random walk escape time, so that for example $\E[\nu_1] = \ET_s/2$.
This gives an alternative proof for \cref{eq:intro-ET}.
It also motivates our conjecture that it is possible to efficiently produce an (approximate) electric flow sample from $s$ using $O(\ET_s)$ steps of a random walk.\footnote{While this seems to follow from the stopping rule, the construction of the stopping rule requires knowledge of the global electric flow, and is hence not feasible to implement.}

We proceed by giving a random walk interpretation of the escape time $\ET_s = \frac{1}{R_s} \sum_x v^2_x d_x$.
Letting $\esc_s = 1 + \max\{t \mid X_t = s\}$ denote the time at which the random walk has left $s$ for the final time (before hitting the sink), we show that
\[
\E[\esc_s]
= \ET_s.
\]
While this seems like an elementary random walk quantity, it seems to not have been studied before.
Apart from its connection to electric flow sampling, we also demonstrate it as a useful quantity for designing random walk algorithms.
Specifically, we design a random walk algorithm for approximating the random walk escape probability or the effective resistance $R_s$ from $s$ to $M$ to a multiplicative factor $\eps$ using $O(\hat\ET_s/\eps^3)$ random walk steps, where $\hat\ET_s \geq \ET_s$ is some given upper bound on the electric hitting time.
This contrasts with a simpler algorithm that requires $O(\HT_s/\eps^2)$ random walk steps.
While trivially $\HT_s \geq \ET_s$, the escape time can be much smaller than the commute time.

\subsection{Trees}
We defined the electric hitting time $\EHT_s$ as the expected number of electric flow samples before the elfs process from source $s$ hits the sink $M$.
As a trivial but important example, consider a path of length $n$ with the source on one end and the sink on the other end.
After a single sample, the source will jump approximately uniformly at random between the source and the sink.
In expectation, this repeatedly halves the distance between the source and the sink, yielding an electric hitting time $\EHT_s \in \Theta(\log n)$.

\begin{figure}[htb]
\centering
\includegraphics[width=.65\textwidth]{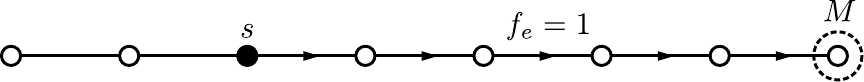}
\caption{$n$-vertex path graph with source $s$ and sink $M$. The electric hitxting time from $s$ to~$M$ is $\EHT_s \in \Theta(\log n)$.}
\label{fig:path}
\end{figure}

Surprisingly, if both endpoints of the line are marked (or equivalently, we consider an $n$-cycle) then bounding the electric hitting time is much more fiddly, yet a similar $\Theta(\log n)$ bound holds.
In fact, and this is our main technical contribution, we show that a similar bound holds for any tree graph.
This contrasts with general graphs, as we show that for instance a complete graph or an expander graph can have electric hitting time $\EHT_s \in \Omega(\poly(n))$.

To state our bound, consider a tree $T$, a source vertex $s$ and a sink $M$.
Without loss of generality, we can assume that $s$ is the root of the tree and $M$ is the set of leaves.
If the tree is unweighted, then we show that
\[
\EHT_s
\leq 2 + \log(R_s) + 2 H(\mu),
\]
with $H(\mu)$ the entropy of the arrival distribution $\mu$ on $M$.
Since $R_s \leq n$ and $H(\mu) \leq \log(|M|)$, this implies that $\EHT_s \in O(\log n)$.
If the tree is weighted, then
\[
\EHT_s
\leq 2 + \log(R_s d_M) + H(\mu) - D(\mu || \sigma_M),
\]
with $d_M$ the total weight of the edges incident on $M$, and $D(\mu || \sigma_M) \geq 0$ the relative entropy between the arrival distribution $\mu$ and the stationary distribution $\sigma_M$ restricted to $M$.
With $w_{\min}$ and $w_{\max}$ the minimum and maximum edge weight, respectively, we can use that $R_s \leq n/w_{\min}$ to get $\EHT_s \in O(\log(n w_{\max}/w_{\min}))$.

As hinted by the path with both endpoints in the sink, the proof for general trees is rather involved.
On a high level, we apply a divide-and-conquer argument.
We use the notion of a Schur complement to show how to contract part of the tree, without affecting the elfs process in the uncontracted part.
This allows us to bound the electric hitting time ``from the bottom up'', i.e., we prove a recurrence relation on the electric hitting time in a tree, based on the electric hitting time in its subtrees.

\paragraph{Schur complement}
On a general graph $G$, the Schur complement on a vertex subset $S$ describes a modified graph $G'$ with vertex set $S$ such that a random walk on $G'$ behaves the same as a random walk on $G$, when considering only the vertices in $S$.
Similarly, the potentials of an electric flow between any pair of vertices in $S$ are the same in $G$ and $G'$.
As a consequence, it seems natural that also the elfs process on $G'$ can be related to that on $G$.
We show that this is the case when there is a single vertex connecting $S$ to its complement $S^c$.
We call such a vertex a ``cut vertex''.
We exploit how a cut vertex implies that, for any source vertex $s \in S$, the electric flow on $S^c$ is always the same (up to a global rescaling).
Since any vertex in a tree graph is a cut vertex, this implies the Schur complement result on trees.

\begin{figure}[htb]
\centering
\includegraphics[width=.9\textwidth]{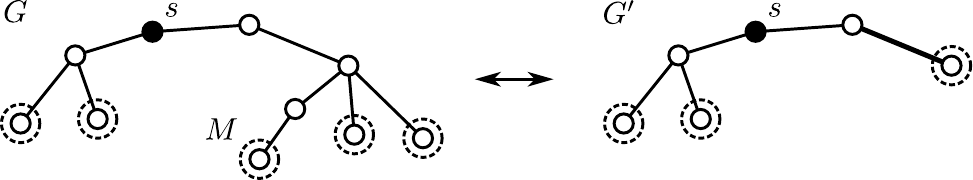}
\caption{Schur complement $G'$ of a tree graph $G$. We relate the number of steps of the elfs process on $G$ with source $s$ and sink $M$ (dotted circles) to the elfs process on $G'$.}
\label{fig:schurcomplement}
\end{figure}

\paragraph{Induction on tree depth}
We then bound the electric hitting time with an inductive argument.
Using the Schur complement, we find that the number of steps that the elfs process spends in a subtree only depends on the structure of that subtree, and the electric flow going from the source $s$ into that subtree.
In particular, there is an easy expression for the number of steps on an edge going directly into the sink.
This is the base case for our inductive argument.
The inductive step then relates the number of steps of the elfs process in a subtree to the number of steps in smaller subtrees of that subtree.
This allows us to bound the electric hitting time by working from the sink, at the bottom of the tree, up to the source, at the root of the tree.

\subsection{Quantum walks}
We showed that electric flow sampling is closely tied to random walks.
However, the original motivation came from the study of \emph{quantum} walks on graphs.
While there are varying definitions, the canonical definition by Szegedy \cite{szegedy2004quantum} describes quantum walks in terms of reflection operators and invariant subspaces.
It was recognized, most notably by Belovs \cite{belovs2013quantum}, that this definition naturally connects with the linear algebraic formulation of electric networks and electric flows.
This was further exploited by Piddock \cite{piddock2019quantum}, who showed that quantum walks can be used very naturally to prepare the quantum flow state
\[
\ket{f}
= \frac{1}{\sqrt{R_s}} \sum_e \frac{f_e}{\sqrt{w_e}} \ket{e}.
\]
Measuring this state returns an edge $e$ with probability $f_e^2/(R_s w_e)$, which corresponds to an electric flow sample.
We show that an $\eps$-approximation of $\ket{f}$ can be prepared using $\tO(\sqrt{\hat\ET_s}/\eps)$ steps of a quantum walk, given an upper bound $\hat\ET_s \geq \ET_s$.
This quadratically improves (and further motivates) the conjectured number of random walk steps required for electric flow sampling.

\paragraph{Quantum elfs process}
We can simulate the elfs process by repeatedly applying the quantum walk algorithm for electric flow sampling.
Ignoring errors, the total number of quantum walks steps required would be
\[
\E\left[ \sum_{j=0}^{\rho-1} \sqrt{\ET_{Y_j}} \right]
\leq \sqrt{2 \HT_s \EHT_s} \leq 2 \HT_s,
\]
where $\{Y_0 = s,Y_1,\dots,Y_\rho \in M\}$ is the original elfs process.
The first inequality follows from combining \cref{eq:sum-ET} with Jensen's inequality, and the second inequality uses $\EHT_s \leq 2 \HT_s$.
This gives an almost quadratic speedup when the electric hitting time $\EHT_s$ is small.
Taking the errors into account, we conjecture that it is possible to $\eps$-approximately simulate the elfs process in running time $\tO(\sqrt{\hat\HT_s \hat\EHT_s}/\poly(\eps))$, where $\hat\HT_s \geq \HT_s$ and $\hat\EHT_s \geq \EHT_s$ are given bounds.
Falling short of proving this, we do prove a weaker bound~of
\begin{equation} \label{eq:intro-QW}
\tO\left( \sqrt{\hat\HT_s} \hat\EHT_s^{3/2} / \eps^2 \right).
\end{equation}
For the particular case of trees we can invoke our earlier bound $\EHT_s \in O(\log( n w_{\max}/w_{\min}))$.
Combined with \cref{eq:intro-QW} this yields a quantum walk algorithm with complexity scaling as $\sqrt{\hat\HT_s}/\eps^2$ (up to polylogarithmic factors) for sampling from a distribution $\eps$-close to the elfs process.

\paragraph{Quantum walk search}
The resulting quantum walk algorithm connects most naturally to the long line of research on quantum walk search, which studies the use of quantum walks for finding distinguished or ``marked'' vertices on a graph.
Building on a long line of work \cite{szegedy2004quantum,magniez2011search,krovi2016quantum,apers2019quantum}, it was recently proven that quantum walks can find such a marked element quadratically faster than the random walk hitting time \cite{ambainis2019quadratic}, provided that we start from the stationary distribution.
This is a nontrivial restriction.
If instead we start from a fixed initial vertex, then a quadratic speedup is only known over the random walk \emph{commute time} \cite{belovs2013quantum,apers2019unified,piddock2019quantum}, which can be much larger than the hitting time.
Our quantum elfs process naturally fits in this quantum walk search framework.
Conceptually, we make two new contributions.

On the one hand, the quantum elfs process is a quantum walk algorithm that finds a marked element that is drawn from the random walk arrival distribution.
This contrasts with most earlier quantum walk algorithms that yield no insight into the quantum walk arrival distribution.
Moreover, it opens up the use of quantum walks to speed up the many classical algorithms that rely on the random walk arrival distribution, including algorithms for random spanning trees \cite{broder1989generating}, semi-supervised learning \cite{zhu2003semi} and linear system solving \cite{chung2015solving}.
We also discuss how our quantum walk algorithm can be generalized to sample from arbitrary target distributions, provided that we have access to a ``local stopping rule'' for that distribution \cite{lovasz1995efficient}.
This complements a recent quantum walk algorithm \cite{bencivenga2020sampling} for speeding up stopping rules when given access to the sampling probabilities.

On the other hand, the runtime of our quantum walk algorithm is closely related to the random walk hitting time.
E.g., on trees our quantum walk algorithm has complexity $\tO(\sqrt{\HT_s})$.
More generally, the conjectured complexity of $\tO(\sqrt{\HT_s\EHT_s})$ is always bounded by $\tO(\HT_s)$.
This contrasts with the aforementioned quantum walk algorithms that start from a single vertex and have a runtime $\tO(\sqrt{\CT_s})$  \cite{belovs2013quantum,apers2019unified,piddock2019quantum}.
Here $\CT_s$ is the random walk \emph{commute time} from $s$ to $M$,
and this is mostly incomparable to the random walk hitting time.

Finally, we also describe a quantum walk algorithm for estimating the escape probability or effective resistance.
This is a quantum counterpart to our random walk algorithm based on the escape time, and has a quadratically better runtime of $\tO(\sqrt{\hat\ET_s}/\eps^{3/2})$, for some given upper bound $\hat\ET_s \geq \ET_s$.
The quantum algorithm refines existing quantum walk algorithms based on the commute time, whose runtime is $\tO(\sqrt{\CT_s}/\eps^{3/2})$ \cite{ito2019approximate,piddock2019quantum}.

\newpage
\part{Elfs}

In this part we analyze the elfs process and its connection to random walks.

\section{Preliminaries}

\subsection{Graphs and electric networks}

Throughout this paper, we denote by $G = (V,E,w)$ a weighted, undirected graph with vertex set $V$, edge set $E$ and nonnegative weights $w:E \to \R_{\geq 0}$.
These weights can typically be thought of as the inverse ``cost'' of an edge.
For convenience, we extend $w$ to a function over all vertex pairs $(x,y) \in V^2$ by setting $w_{xy} = 0$ if $(x,y) \notin E$.

We can associate an electric network to $G$ by interpreting any edge $e \in E$ as a resistor with conductance $w_e$.
If $(x,y) \notin E$ then it has conductance $w_{xy} = 0$.
For any \emph{source}~$s \in V$ and \emph{sink} $M \subset V$, consider the unit electric flow $f^{s,M} = f$ from $s$ to $M$.\footnote{We drop the superscripts on $f^{s,M}$ when the source $s$ and sink $M$ are clear from context.}
This is the unique unit flow $f:E \to \R$ from $s$ to $M$ that minimizes the dissipated energy
\[
\cE(f)
= \sum_e \frac{f_e^2}{w_e}.
\]
The resulting energy of the electric flow is called the effective resistance $R_s = \cE(f^{s,M})$.
Alternatively, the electric flow is the unique potential flow from $s$ to $M$.
In other words, there exists a \emph{potential} vector $v^{s,M} = v$ over $V$ such that for all $x,y$ it holds that
\[
f_{xy}
= w_{xy} (v_x - v_y).
\]
While this potential vector is not unique, we fix it by demanding $v_y = 0$ for $y \in M$.
In this case we have that $v_s = R_s$.

We denote by $L = D - A$ the Laplacian associated to $G$, where $A$ is the weighted adjacency matrix with $(A)_{xy} = w_{xy}$ and $D$ is the diagonal degree matrix with $(D)_{xx} = \sum_y w_{xy}$.
We can describe the potential vector $v$ using the Laplacian by
\[
v_x
= (L_{UU}^{-1})_{sx}.
\]
where $L_{UU}$ denotes the principal submatrix with rows and columns indexed by vertices in $U = V \backslash M$.
The effective resistance is then $R_s = v_s = (L_{UU}^{-1})_{ss}$.

\subsection{Markov chains and random walks}

Consider a Markov chain $\{X_0,X_1,X_2,\dots\}$ with states $X_t \in V$.
Its dynamics are determined by transition probabilities $P_{xy} = \Pr(X_{t+1} = y \mid X_t = x)$.
As a consequence, if the distribution of $X_t$ is described by row vector $p_t$, then $X_{t+1}$ is distributed according to the distribution $p_{t+1} = p_t P$.
A \emph{random walk} on a graph $G$ is described by a Markov chain with transition matrix $P=D^{-1}A=I-D^{-1}L$.
It turns out that any \emph{reversible} Markov chain can be written as a random walk on an appropriately weighted graph.

A Markov chain that is \emph{absorbing} in a sink $M \subseteq V$ has a transition matrix of the form
\[
P = \left( \begin{array}{cc}
P_{UU} & P_{UM} \\
0 & I_{MM} \\
\end{array} \right),
\]
where we have split the state space into sink vertices ($M$) and remaining vertices ($U = V \backslash M$).

\paragraph{Hitting time.}
For a general Markov chain, the hitting time $\tau_M$ of sink $M$ is the expected time at which the chain hits a vertex in $M$.
Starting from a vertex $s$, we denote its expectation value $\E_s(\tau)$ by $\HT_s$.
We have that (see e.g.~\cite{kemeny1983finite})
\[
\HT_s
= \sum_x (I_{UU} - P_{UU})^{-1}_{sx}.
\]
For a random walk from $s$ this becomes
\[
\HT_s
= \sum_x (I_{UU} - P_{UU})^{-1}_{sx}
= \sum_x (L_{UU}^{-1} D)_{sx}
= \sum_x v_x d_x,
\]
where we recall that $v_x$ is the voltage at $x$ in the unit electric flow from $s$ to $M$.

\paragraph{Arrival distribution.}
We will also be interested in the distribution with which the random walk hits $M$.
We call this the \emph{arrival distribution} $\mu$.
For a Markov chain that starts from~$s$ and that is absorbing in $M$, this is given by the infinite time distribution $\mu_y = \lim_{n\rightarrow \infty} (P^n)_{sy}$.
Assuming that the hitting time is finite from any initial vertex (and so $\norm{P_{UU}} < 1$), we can rewrite
\[
\lim_{t \to \infty} P^t
= \lim_{t \to \infty} \left( \begin{array}{cc}
P_{UU}^t & \sum_{i=0}^{t-1} P_{UU}^t P_{UM} \\
0 & I_{MM} \\
\end{array}\right)
= \left( \begin{array}{cc}
0 & (I_{UU}-P_{UU})^{-1} P_{UM} \\
0 & I_{MM} \\
\end{array}
\right),
\]
and so
\begin{equation} \label{eq:arrival-distr}
\mu_y
= \big((I_{UU}-P_{UU})^{-1} P_{UM}\big)_{sy}.
\end{equation}
For a random walk this is
\begin{equation} \label{eq:RW-arrival-distr}
\mu_y
= \big((I_{UU}-P_{UU})^{-1} P_{UM}\big)_{sy}
= \big((L_{UU}^{-1} D_{UU}) (-D_{UU}^{-1}L_{UM}) \big)_{sy}
= -(L_{UU}^{-1} L_{UM})_{sy}.
\end{equation}
Recalling that $(L_{UU}^{-1})_{sx} = v_x$, this becomes
\[
\mu_y
= \sum_x v_x w_{xy}
= \sum_x f_{xy}.
\]
I.e., the probability of arriving in sink vertex $y$ equals the electric flow going into $y$.

\paragraph{Other useful identities.}

We summarize some additional identities that relate properties of random walks and electric flows.
Proofs of these identities can be found in e.g.~\cite{lyons2017probability}.
For a random walk from $s$ that is absorbing in $M$, let $S_{xy}$ denote the expected number of traversals from $x$ to $y$.
Then
\begin{equation} \label{eq:exp-traversals}
\E_s[S_{xy}]
= v_x w_{xy}.
\end{equation}
With $T_x$ the expected number of visits to $x$, this implies that
\begin{equation} \label{eq:exp-visits}
\E_s[T_x]
= v_x d_x.
\end{equation}
This generalizes the expression for the random walk arrival distribution.

Now consider a non-absorbing random walk starting from $s$.
We denote by $\tau^+_s$ the expected \emph{return time} to $s$.
By Kac's lemma this is equal to
\begin{equation} \label{eq:exp-return}
\E_s[\tau^+_s]
= \frac{1}{\pi(s)}
= \frac{d_s}{\sum_x d_x}.
\end{equation}
Finally, we also have that
\begin{equation} \label{eq:return-prob}
\Pr_s[\tau_s^+ > \tau_M]
= \frac{1}{R_s d_s}
\quad\text{ and }\quad
\Pr_x[\tau_s < \tau_M]
= \frac{v_x}{R_s}.
\end{equation}

\section{Electric flow sampling process}

We begin our analysis of the elfs process, which we recall below.

\begin{algorithm}[H]
\caption*{\bf Electric flow sampling (elfs) process} \label{alg:elfsrep}
\begin{algorithmic}
\vspace{1mm}
\State
From an initial source vertex $s$ and fixed sink $M \subseteq V$, repeat the following:
\begin{enumerate}
\item
Let $f$ denote the unit electric flow from source $s$ to sink $M$.\newline
Sample an edge $e$ with probability proportional to $f_e^2/w_e$.
\item
Pick a random endpoint $x$ from $e$ and set the source $s = x$.
\end{enumerate}
The process ends when step 2.~picks a vertex $x \in M$.
\end{algorithmic}
\end{algorithm}

In the following, we give a closed form algebraic expression for the transition matrix of the elfs process.
From this expression, we find expressions related to the arrival distribution and electric hitting time.

\subsection{Algebraic form and arrival distribution}

For a fixed sink $M$, the elfs process describes a Markov chain on the vertex set $V$.
We let $Q$ denote its transition matrix.
Starting from a source $s$, the transition probability $Q_{sx}$ denotes the probability of picking $x$ in step 2.
In the following lemma we give a closed and useful expression for~$Q$.
We use shorthand $A * B$ for the entrywise product of two matrices, and introduce voltage matrix $V$ and diagonal resistance matrix $R$ defined by
\[
(V)_{xy}
= v^x_y
\quad \text{ and } \quad
(R)_{xx}
= R_x.
\]

\begin{lemma} \label{lem:algebraic-Q}
The transition matrix of the elfs process with sink $M$ is
\[
Q
= \left(\begin{matrix}
	I_{UU}-\frac{1}{2}R^{-1}(V*V)L_{UU} & - \frac{1}{2}R^{-1}(V*V)L_{UM}\\
	0 & I_{MM}
\end{matrix}\right). \]
In particular, with $v$ the voltages of the unit electric flow from $s$ to $M$,
\begin{equation} \label{eq:Qsx}
Q_{sx}
= \frac{1}{2 R_s} \sum_y (v_x - v_y)^2 w_{xy}.
\end{equation}
\end{lemma}
\begin{proof}
The probability of picking an edge $(x,y)$ when sampling from the electric flow from $s$ to $M$ is
\[
Q_{s,(x,y)}
= \frac{1}{R_s} \frac{(f_{xy})^2}{w_{xy}}
= \frac{1}{R_s} (v_x-v_y)^2 w_{xy}.
\]
The probability of picking a vertex $x$ is hence
\begin{align*}
Q_{sx}
= \sum_y \frac{Q_{s,(x,y)}}{2}
&= \sum_y \frac{(v_x-v_y)^2 w_{xy}}{2R_s} \\
&= 2 v_x \sum_y \frac{(v_x-v_y)w_{xy}}{2 R_s}
	- \sum_y \frac{(v_x^2-v_y^2)w_{xy}}{2 R_s} \\
&= \frac{v_x}{R_s} \delta_{sx} - \sum_y \frac{(v_x^2 - v_y^2) w_{xy}}{2 R_s}
= \delta_{sx} - \frac{1}{2 R_s} \sum_y v_y^2 L_{yx}.
\end{align*}
In the last equality we used that $v_s = R_s$.
In the third equality we used flow conservation and the fact that $v_x = 0$ for $x \in M$ to rewrite
\[
v_x \left(\sum_y (v_x - v_y) w_{xy}\right)
= v_x \left(\sum_y f_{xy}\right)
= v_x \left( \delta_{sx} - \sum_{z \in M} q_z \delta_{xz} \right)
= v_x \delta_{sx},
\]
with $q_z$ the flow going into vertex $z \in M$.
The final expression for $Q_{sx}$ exactly equals the claimed expression.
\end{proof}

From this expression we can easily prove that the elfs arrival distribution equals the random walk arrival distribution.

\begin{corollary} \label{cor:arrival-distribution}
The arrival distribution of the elfs process from initial source $s$ to sink $M$ is equal to the arrival distribution of a random walk from $s$ to $M$.
\end{corollary}
\begin{proof}
Plugging the expression for the elfs transition matrix $Q$ into expression \cref{eq:arrival-distr} for the Markov chain arrival distribution, we get that the elfs arrival distribution equals $q_y
= ((I_{UU}-Q_{UU})^{-1} Q_{UM})_{sy}$.
By \cref{lem:algebraic-Q} we can rewrite
\begin{align*}
(I_{UU}-Q_{UU})^{-1} Q_{UM}
&= \left(2 L_{UU}^{-1}(V*V)^{-1}R\right) \left(-\frac{1}{2}R^{-1}(V*V)L_{UM}\right)
= -L_{UU}^{-1}L_{UM},
\end{align*}
and so $q_y = (-L_{UU}^{-1}L_{UM})_{sy}$.
This matches the standard random walk arrival distribution in \cref{eq:RW-arrival-distr}.
\end{proof}

In fact, even the edge from which the elfs process arrives in the sink is distributed equally to the edge with which a random walk arrives in the sink (whose probability equals the flow along that edge).
This is easily seen by considering a modified graph $\hat G$ obtained by replacing any sink vertex $y$ by a separate vertex $u_{xy}$ for every incoming edge $(x,y)$.
The probability of arriving in $u_{xy}$ (and hence of taking the edge $(x,y)$) is then $f_{xy}$ for both the random walk and the elfs process.
Given that both processes behave exactly the same on $G$ and $\hat G$, this must be the case on the original graph as well.

\subsection{Electric hitting time}

Next we consider the \emph{electric hitting time}.
From an initial source vertex $s$ and a fixed sink $M$, the electric hitting time $\EHT_s$ equals the expected number of samples before the elfs process hits the sink.
We discuss the electric hitting time on illustrative examples such as the path graph, the binary tree and the complete graph.
Then we analyze the electric hitting time using the algebraic expression from \cref{lem:algebraic-Q}.
We show that a single electric flow sample reduces in expectation the (random walk) hitting time by a certain amount.
This proves for instance that $\EHT_s \leq 2\HT_s$.
In the later \cref{part:trees} we prove that the electric hitting on trees is logarithmic in the graph size and weights.

\subsubsection{Examples} \label{sec:examples}

In the following we analyze the electric hitting time on some canonical (unweigthed) graph examples.

\paragraph{Path graph.}
A simple yet key example is that of the elfs process on a path graph with one sink vertex (see \cref{fig:path} in the intro).
The unit electric flow from a source vertex $s$ to the sink $M$ takes value $f_{xy} = 1$ on each of the edges between $s$ and $M$.

A single electric flow sample from $s$ will return an approximately uniform vertex between $s$ and $M$.
This will repeatedly half the distance from $s$ to $M$, and for a path of length $n$ this yields an electric hitting time $\EHT_s \in \Theta(\log n)$.

As trivial as this example may seem, the situation becomes surprisingly more complicated when the sink includes \emph{both} ends of the path.\footnote{Equivalently, consider a single sink vertex on a \emph{cycle} graph.}
It follows from our later bound for the electric hitting time on trees in \cref{thm:elfstreebound} that also this example has electric hitting time $O(\log n)$.

\paragraph{Complete graph.}
Another insightful exercise is to analyze the electric hitting time on the complete graph.
Consider an initial source vertex $s$ and a sink $M \subseteq V$.
By explicit calculation we find the probability $p$ of hitting a sink vertex after a single electric flow sample from $s$.
Thanks to the symmetry of the complete graph, this probability is independent of $s$.
Hence, the elfs process effectively corresponds to a repeated Bernoulli experiment with success probability $p$.
The electric hitting time is then simply the number of trials until success, which is $\EHT_s = 1/p$.

\begin{figure}[htb]
\centering
\includegraphics[width=.35\textwidth]{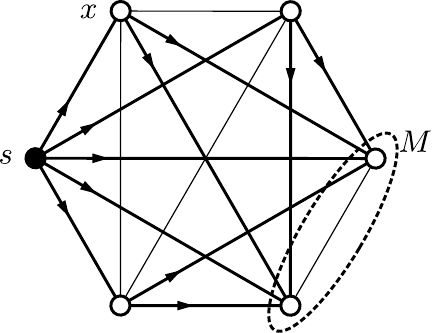}
\caption{$n$-vertex complete graph with source $s$ and sink $M$. The electric hitting time from $s$ to $M$ is $\Theta(|M|(1-|M|/n))$.}
\label{fig:complete}
\end{figure}

It remains to calculate the probability $p$ that an electric flow sample from $s$ returns a sink vertex.
To this end, one can check that the unit electric flow from $s$ to $M$ is determined by the voltages $v_s = R_s = (|M|+1)/(n|M|)$ and $v_x = 1/(n|M|)$ for $x \notin M$.
After some straightforward calculation this yields a success probability $p$ that is
\[
p
= \frac{1}{2 R_s} \sum_{y \in M} \sum_{x \notin M} f_{xy}^2
= \frac{1}{2} \left( \frac{1}{|M|+1} + \frac{|M|}{n} \right).
\]
The electric hitting time is hence $\EHT_s = 1/p \in \Theta(\min\{|M|,n/|M|\})$.
In contrast to the random walk hitting time, this shows that adding more vertices to the sink can actually \emph{increase} the electric hitting time.

\paragraph{Expander graph.}
Using a more qualitative argument, we see that a similar bound holds for an expander graph.
We use that for a $d$-regular expander it holds that $R_{x,M} \in \Theta(1/d)$, $\HT_x \in \Theta(n/|M|)$ and $\CT_x \in \Theta(n)$.

It is now easy to see that $\EHT_s \in O(\min\{|M|,n/|M|\})$. The $O(n/|M|)$ bound follows from the fact that $\EHT_s \in O(\HT_s) \in O(n/|M|)$.
The $O(|M|)$ bound follows from the fact that the flow going into $M$ can only spread over $O(d |M|)$ edges, and that the effective resistance is always $\Theta(1/d)$.

\subsubsection{What happens to the classical hitting time?}
We can use the classical hitting time as a measure of ``distance'' from the sink. After one step of electric flow sampling, the expected hitting time of the classical random walk is (using $\ket{\ones}$ to denote the all ones vector)
\begin{align}\bra{s} Q_{UU}VD \ket{\ones} &=\bra{s} VD \ket{\ones}-\frac{1}{2}\bra{s} R^{-1}(V*V) L_{UU} VD \ket{\ones}\\
	&=\bra{s} VD \ket{\ones}-\frac{1}{2}\bra{s} R^{-1}(V*V)D \ket{\ones} \\
	&= \HT_s - \frac{\ET_s}{2}, \label{eq:HTdecrease} \end{align}
where we introduced the shorthand
\begin{equation} \label{eq:def-ETs}
\ET_s
= \bra{s} R^{-1}(V*V)D \ket{\ones}
= \frac{1}{R_s}\sum_x v_x^2 d_x.
\end{equation}
We will call this quantity the ``random walk escape time'' from $s$, a name that we will justify in the next section.
So, in expectation, a single elf sample knocks off $\ET_s/2$ from the expected hitting time.
The following bounds will be useful:
\begin{equation} 1 \leq R_s d_s\le \ET_s \le \HT_s 
	\label{eq:xbound}
\end{equation}
\begin{proof}
	For the right hand side, we use that $v_x \le R_s$ for all $x$ to show that \item $\frac{1}{R_s} \sum_x v_x^2 d_x \le \sum_x v_x d_x = \HT_s$.
For the left hand side we have
\[
\frac{1}{R_s} \sum_x v_x^2 d_x
\ge \frac{1}{R_s} v_s^2 d_s
= R_s d_s \ge 1. \qedhere
\]
\end{proof}

With $\{Y_0 = s,Y_1,\dots,Y_\rho \in M\}$ the elfs process, it follows (and we can check directly) that:
\begin{equation}
	\label{eq:Esumx}
	\EHT_s
	= \E[\rho]
	\leq \E\left[\sum_{j=0}^{\rho-1} \ET_{Y_j}\right]
	= 2 \HT_s,
\end{equation}
so that the electric hitting time is at most twice the random walk hitting time.

\subsubsection{What happens to the electric potential?}
Alternatively, for any $s \notin M$, we can use the electric potential $\bra{s}V$ to measure our progress.
Take an arbitrary initial vertex $x \notin M$ and let $y$ be the sampled vertex, then we have that
\[
\E_x[V_{sy}]
= \E_x[\braket{y|V|s}]
= \braket{x|Q_{UU}V|s}.
\]
Using the expression $Q_{UU} = I_{UU} - \frac{1}{2} R^{-1} (V \ast V) L_{UU}$ this becomes
\[
\E_x[v^s_y]
= \braket{x|V|s} - \frac{\braket{x|V|s}^2}{2R_x}
= v^s_x \left( 1 - \frac{v^s_x}{2R_x} \right).
\]
For $x = s$ this shows that after a single sample the potential halves in expectation.
Using that $v^s_x = v^x_s$ and $v^s_x = R_x \Pr_s[\tau_x < \tau_M]$ we can cleanly rewrite the expression as
\[
\E_x[v^s_y]
= v^s_x \left( 1 - \frac{1}{2} \Pr_s[\tau_x < \tau_M] \right).
\]

\section{Random walks}

In this section we show that there exists a coupling between the elfs process and a random walk from source $s$ to sink $M$.
This gives more insightful proofs for the facts that (i) the elfs process has the same arrival distribution as a random walk, and (ii) a single electric flow sample reduces the hitting time by $\ET_s$ in expectation.
Then, we give an interpretation of this same quantity $\ET_s$ as the random walk \emph{escape time}, and use it to design random walk algorithms for estimating effective resistances.

\subsection{Random walk coupling} \label{sec:RW-coupling}

Here we construct a coupling between a random walk from $s$ to $M$ and the elfs process, described by respective Markov chains
\[
\{X_0=s, X_1,\dots,X_\tau \in M\}
\quad \text{ and } \quad
\{Y_0=s,Y_1,\dots,Y_\rho \in M\}.
\]
Note that $\E[\tau] = \HT_s$ and $\E[\rho] = \EHT_s$.
We also proved (and will reprove) that both walks have the same arrival distribution, and so $\Pr[Y_\rho = y] = \Pr[X_\tau = y]$ for all $y$.

The coupling is based on a simple random walk stopping rule $\nu$ for which $X_\nu = Y_1$.
For a random walk starting from $s$, the rule $\nu$ is the following: when visiting a vertex $x$, stop with probability
\begin{equation} \label{eq:stopping-prob}
p_x
= \frac{w_{xx'}}{w_{xx'} + d_x},
\quad \text{ with }
w_{xx'}
= \sum_y \left( \frac{v_x-v_y}{v_x} \right)^2 w_{xy},
\end{equation}
and otherwise continue.
This is equivalent to appending to every vertex $x \in U$ a new sink vertex $x'$ through an edge with weight $w_{xx'}$.
The main result of this section is the following lemma, whose proof we postpone to the end of the section.
\begin{lemma}[Coupling lemma] \label{lem:coupling}
The probability that a random walk from $s$ with stopping rule~$\nu$ stops at vertex $x$ is
\[
\Pr_s[X_\nu = x]
= \Pr_s[Y_1 = x] = Q_{s,x}.
\]
Moreover, the expected length of the stopping rule is $\E_s[\nu] = \ET_s/2$.
\end{lemma}
By repeatedly applying this stopping rule (until the random walk stops at a vertex in~$M$), we obtain stopping times $\nu_1,\nu_2,\dots,\nu_\rho = \tau$ such that
\[
\left( \{X_0=s, X_1,\dots,X_\tau \in M\},\,
\{Y_0=s,Y_1 = X_{\nu_1},\dots,Y_\rho = X_{\nu_\rho} \in M\} \right)
\]
describes a coupling between the random walk and the elfs process.
We illustrate this coupling in \cref{fig:coupling}.

\begin{figure}[htb]
\centering
\includegraphics[width=.6\textwidth]{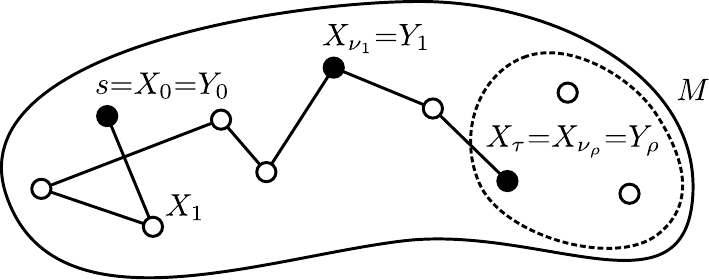}
\caption{(identical to Figure \ref{fig:intro-coupling}) Coupling between a random walk $\{X_0 = s,X_1,\dots,X_\tau \in M\}$ and an elfs process $\{Y_0 = s,Y_1,\dots,Y_\rho \in M\}$ through a stopping rule $\nu$.
From $s$ this ensures that $X_\nu = Y_1$.}
\label{fig:coupling}
\end{figure}

The coupling implies directly that the elfs process has the same arrival distribution as a random walk, reproving \cref{cor:arrival-distribution}.
The coupling also gives an alternative proof for \cref{eq:Esumx}, which effectively states that
\[
\E\left[ \sum_{j=1}^\rho (\nu_j-\nu_{j-1}) \right]
= \E[\nu_\rho]
= \E[\tau] = \HT_s.
\]
We now give the proof of the coupling lemma.
\begin{proof}[Proof of \cref{lem:coupling} (coupling lemma)]
We consider the equivalent formulation of the stopping rule using a random walk on a modified graph $G'$, which is obtained by adding to each non-sink vertex $x \in U$ an additional sink vertex $x'$ with an edge $(x,x')$ of weight $w_{xx'}$ (as in \cref{eq:stopping-prob}).
Let $U'$ denote this new set of vertices.
If we run a random walk on $G'$ until it hits a vertex in the new sink $U' \cup M$, then it will stop at a vertex $x' \in U'$ with arrival probability $\Pr[X_\nu = x]$ (and it will stop at a vertex $y \in M$ with probability $\Pr[X_\nu = y]$).

\begin{figure}[htb]
\centering
\includegraphics[width=.6\textwidth]{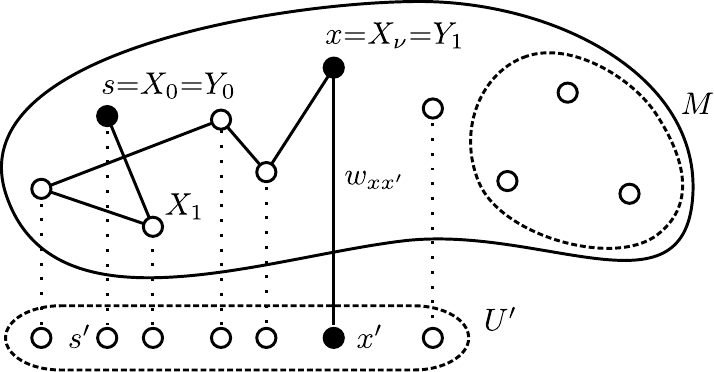}
\caption{Modified graph $G'$. The stopping rule $\nu$ is equivalent to a random walk from $s$ hitting a sink vertex in $U' \cup M$.}
\label{fig:mod-graph}
\end{figure}

Let $f'$ denote the unit electric flow on $G'$ from $s$ to $U' \cup M$.
We argue that this flow is determined by the voltages
\[
v'_x
= v_x^2/(2R_s)
\]
for $x \in U$ and $v'_x = 0$ for $x \in U' \cup M$.
Recalling that the voltages are the \emph{unique} potentials that induce a unit flow from $s$ to $U' \cup M$, it suffices to prove that these voltages induce such a unit flow.
This follows from direct computation: for any $x \in U$ we have that the outgoing flow is
\begin{align*}
\sum_y f'_{xy} + f'_{xx'}
&= \sum_y (v'_x-v'_y) w_{xy} + v'_x w_{xx'} \\
&= \sum_y \frac{v_x^2 - v_y^2}{2R_s} w_{xy}
    + \frac{v_x^2}{2R_s} \sum_y \left( \frac{v_x-v_y}{v_x} \right)^2 w_{xy} \\
&= \frac{v_x}{R_s} \sum_y (v_x - v_y) w_{xy}
= \frac{v_x}{R_s} \delta_{sx} = \delta_{sx}.
\end{align*}
For $y \in M$ the incoming flow is
\[
\sum_x f'_{xy}
= \sum_{x \in U} v'_x w_{xy}
= \sum_{x \in U} \frac{v_x^2}{2R_s} w_{xy}.
\]
By \cref{eq:Qsx} this equals $Q_{sy}$.
As the incoming flow equals the arrival probability $\Pr[X_\nu = y]$, we have that $\Pr[X_\nu = y] = Q_{sy}$.
For $x' \in U'$ it is
\[
f'_{xx'}
= v'_x w'_{xx'}
= \frac{v_x^2}{2R_s} \sum_y \left( \frac{v_x-v_y}{v_x} \right)^2 w_{xy}
= \frac{1}{2R_s} \sum_y (v_x-v_y)^2 w_{xy}.
\]
By \cref{eq:Qsx} this equals $Q_{sx}$, and hence $\Pr[X_\nu = x] = Q_{sx}$.
This proves the first statement of the lemma.

To prove the second statement, note that the expected length of the stopping rule $\nu$ equals the expected number of edges traversed by the random walk between the ``original'' vertices $U \cup M$ (i.e., excluding edges of the form $(x,x')$), up until the random walk hits $U' \cup M$.
By \cref{eq:exp-traversals}, the expected number of times that the random walk crosses such an edge $(x,y)$ is $v'_x w_{xy}$.
Summing this over all edges in the original graph we get
\[
\sum_{x,y \in U \cup M} v'_x w_{xy}
= \sum_{x \in U} v'_x d_x
= \frac{1}{2 R_s} \sum_{x \in U} v_x^2 d_x
= \frac{\ET_s}{2}. \qedhere
\] 
\end{proof}

There exists a similar stopping rule for sampling \emph{edges}.
A \emph{lazy} random walk $\{\tilde X_0=s,\tilde X_1,\dots,\tilde X_\tau \in M\}$ from $s$ to $M$ can be described as follows: from the current vertex $\tilde X_i = x$, pick a random outgoing edge $Z_i = e = (x,y)$ with probability $w_{xy}/d_x$, and let $\tilde X_{i+1}$ be a random endpoint of $Z_i$.
Now consider the edge stopping rule $\mu$: when picking edge $e = (x,y)$, stop with probability
\[
p_{xy}
= \frac{(v_x-v_y)^2}{v_x^2 + v_y^2}.
\]
We have the following edge coupling lemma, the proof of which we defer to \cref{app:edge-coupling}.
\begin{restatable}[Edge coupling lemma]{lemma}{edgecoupling} \label{lem:edge-coupling}
The probability that a random walk from $s$ with stopping rule~$\mu$ stops at edge~$e$ is
\[
\Pr_s[Z_\mu = e]
= f_e^2/(R_s w_e).
\]
\end{restatable}

\subsection{Random walk escape time}

For a random walk $\{X_0=s,X_1,\dots,X_\tau \in M\}$ from $s$ to $M$, consider the quantity $\esc = 1 + \max\{t \mid X_t = s\}$, which denotes the time at which the random walk has left $s$ for the final time.
In this section we prove the following lemma, with $v$ the voltage in the unit electric flow from $s$ to $M$.
\begin{lemma}[escape time identity] \label{lem:escape-time}
$\E_s(\esc) = \ET_s = \frac{1}{R_s} \sum_x v_x^2 d_x$.
\end{lemma}
\begin{proof}
We prove this identity by combining two existing identities.
First, on a random walk from $s$ to $M$, the expected number of visits to a vertex $x$ is $v_x d_x$ (\cref{eq:exp-visits}).
Second, the probability that a walk starting at $x$ hits $s$ before $M$ is $\Pr_x[\tau_s<\tau_M] = v_x/R_{s}$ (\cref{eq:return-prob}).
Hence, a fraction $v_x/R_s$ of the visits to $x$ will return to $s$ before $M$, and this gives a total of $v_x^2 d_x/R_s$ visits to $x$ before the walk hits $s$ for the final time.
Taking the sum over all $x$ then gives the expected escape time $\E_s[\esc] = \frac{1}{R_s} \sum_x v_x^2 d_x$.

More rigorously, we expand
\[
\E_s[\esc]
= \E_s\left[ \sum_{t=0}^\infty \ind(t < \esc) \right]
= \sum_{t=0}^\infty \Pr_s[t < \esc].
\]
By rewriting
\[
\Pr_s[t < \esc]
= \sum_x \Pr_s[X_t = x] \Pr_x[\tau_s < \tau_M]
= \sum_x \Pr_s[X_t = x] \frac{v_x}{R_{s}},
\]
this becomes
\[
\E_s[\esc]
= \sum_{t=0}^\infty \sum_x \Pr_s[X_t = x] \frac{v_x}{R_{s}}
= \sum_x \frac{v_x}{R_{s}} \sum_{t=0}^\infty \Pr_s[X_t = x].
\]
By \cref{eq:exp-visits} we have that $\sum_{t=0}^\infty \Pr_s[X_t = x] = \E_s\left[\sum_{t=0}^\infty \ind[X_t = x]\right] = v_x d_x$, and this gives the claimed expression.
\end{proof}
\noindent
For the interested reader, we give in \cref{app:escape-time-doob} an alternative proof of the escape time identity using an elegant tool called the \emph{Doob transform}.

The escape time identity justifies the name \emph{random walk escape time} for the quantity $\ET_s$, which we encountered previously.
Note that the escape time from $s$ clearly satisfies the inequalities
\[
1
\le |\{ t \: | \: X_t = s \}|
\le \esc_s
\le \tau,
\]
with $\tau$ the hitting time from $s$ to $M$.
Taking expectations, and using that $\E_s[|\{ t \: | \: X_t = s \}|] = R_s d_s$ (by \cref{eq:return-prob}), we find an alternative proof of the bounds in \cref{eq:xbound} which state that
\[
1
\le R_s d_s
\le \ET_s
\le \HT_s.
\]

\subsection{Local rules and general target distributions} \label{sec:stopping-rules}

Here we reframe the past few results in the theory of stopping rules as described by Lovász and Winkler \cite{lovasz1995efficient}.

\paragraph*{Local rule.}
Following the terminology of \cite{lovasz1995efficient}, the stopping rule in \eqref{eq:stopping-prob} is an example of a \emph{local rule}.
For some fixed initial distribution $\sigma$ and target distribution $\gamma$, a local rule associates to each vertex $x$ a stopping probability $p_x$.
This describes a stopping rule from $\sigma$ to $\gamma$ if, when the walk starts from $\sigma$, the walk stops according to $\gamma$.
They show that for any choice of $\sigma$ and $\gamma$ there exists such a choice of probabilities $\{p_x\}$.
In fact, it is not hard to see that our stopping rule corresponds to the precise local rule construction in \cite{lovasz1995efficient} that is based on ``exit frequencies'' (see their paper for a precise definition).\footnote{They set $p_x = \gamma_x/(\gamma_x + \mathrm{EF}_x)$ where $\mathrm{EF}_x$ is the exit frequency at $x$. If $\sigma$ is the indicator on the source~$s$, and $\gamma$ the elfs distribution, then one can show that the exit frequency $\mathrm{EF}_x = v_x^2 d_x/(2R_s)$ and their construction reproduces our rule from \eqref{eq:stopping-prob}.}

\paragraph*{Mean-optimal rules.}
Another interesting result in \cite{lovasz1995efficient} concerns ``mean-optimal'' stopping rules.
A mean-optimal stopping rule is a stopping rule from $\sigma$ to $\gamma$ whose expected length (i.e., the expected number of steps before stopping) is minimal over all possible (not necessarily local) stopping rules.
The expected length of such a mean-optimal rule is defined as the \emph{access time} (or generalized hitting time) $\HT_{\sigma,\gamma}$ from $\sigma$ to $\gamma$.
In \cite[Theorem 5.1]{lovasz1995efficient} they show that \emph{a stopping rule is mean-optimal if and only if it has a halting state}, where a halting state is a vertex such that the walk \emph{always} stops when visiting that vertex.
Now note that any sink vertex in our local rule \eqref{eq:stopping-prob} is a halting state, and hence our stopping rule is mean-optimal.
Since our rule has expected length $\ET_s/2$ (\cref{lem:coupling}), this proves that the access time from $s$ to the elfs distribution equals $\ET_s/2$, which is half the escape time (\cref{lem:escape-time}).

\paragraph*{General arrival distributions.}
We already mentioned that for any $\sigma$ and $\gamma$ there exist local stopping probabilities $\{p_x\}$ that form a stopping rule from $\sigma$ to $\gamma$.
We can reinterpret such stopping probabilities as in \cref{fig:mod-graph}: we connect every vertex $x$ to a new sink vertex $x'$ via an edge with weight $w_{xx'} = p_x d_x/(1-p_x)$, so that a random walk jumps from $x$ to the sink vertex $x'$ with probability $w_{xx'}/(w_{xx'}+d_x) = p_x$  (if $p_x=1$ then we let $x$ be the sink vertex).
Let $G'$ be the resulting graph, and let $M$ denote its sink. 
If $\sigma$ is the indicator on a source vertex $s$, then the target distribution $\gamma$ is precisely the arrival distribution of a random walk on $G'$ from $s$ to the sink $M$.
This small observation implies that we can massage any target distribution $\gamma$ into a random walk arrival distribution, and so our preceding discussion carries over to this generalized setting as well.
In our discussion on quantum walks this will imply that there exist quantum walk algorithms for sampling from arbitrary target distributions (provided we are given a local stopping rule).

\subsection{Approximating effective resistances}

In this section we illustrate the use of the escape time in designing random walk algorithms.
We describe a classical algorithm that approximates the escape probability or effective resistance $R_s$ using $O(\hat\ET_s)$ random walk steps from $s$, for some give upper bound $\hat\ET_s \geq \ET_s$.
We contrast this with a simpler algorithm that requires $O(\HT_s)$ random walk steps (but has a better dependency on the approximation factor).
The use of random walks to estimate graph parameters is an interesting area on itself, and we refer the interested reader to for instance \cite{benjamini2006waiting,cooper2014estimating,ben2019estimating,peng2021local}.
In the second half of the paper, we will discuss how these random walk algorithms parallel existing algorithms based on \emph{quantum} walks \cite{ito2019approximate,piddock2019quantum}.

At the end of this section, we formulate a conjecture stating that it is possible to produce an electric flow sample (approximately) from $s$ using $O(\ET_s)$ random walk steps.

\subsubsection{Approximating the effective resistance} \label{sec:RW-eff-res}

Here we describe two algorithms for estimating the random walk escape probability $p_s$ from a source $s$ and sink $M$, which is defined as
\[
p_s
\coloneqq \Pr[\tau_M < \tau^+_s],
\]
with $\tau_M$ the hitting time from $s$ to $M$ and $\tau^+_s$ the return time from $s$ back to $s$.
Recalling \cref{eq:return-prob}, this relates to the effective resistance $R_s$ through the relation
\[
p_s
= \frac{1}{d_s R_s},
\]
so that equivalently the algorithms estimate the effective resistance (when given $d_s$).
Both algorithms are based on random walks, and can be implemented space efficiently.

\paragraph{Hitting time algorithm.}

The following algorithm has expected runtime $O(\HT_s/\eps^2)$ and it returns an $\eps$-approximation of $R_s d_s$ (and hence the escape probability) with constant probability.

\begin{algorithm}[H]
\caption*{\bf Hitting time algorithm for estimating $R_s d_s$} \label{alg:HT-R}
\begin{algorithmic}
\vspace{1mm}
\State
Repeat for $i = 1,2,\dots,k \in O(1/\eps^2)$:
\begin{enumerate}
\item
From source $s$, run a random walk until you hit $M$.
Let $S_i$ be the number of visits to $s$ during the walk.
\end{enumerate}
Return $\frac{1}{k} \sum_{i=1}^k S_i$.
\end{algorithmic}
\end{algorithm}

We now argue that, with constant probability, the algorithm returns an approximation of $R_s d_s$ with multiplicative error $\eps$.
Using that $\Pr[\tau_M < \tau_s^+] = 1/(R_s d_s) = p_s$, we have that the random variable $S_i$ corresponds to the number of trials until success of a Bernoulli variable with success probability $p_s$.
For any $i$ we then have that $\E_s(S_i)=1/p_s=R_s d_s$ and $\E_s(S_i^2)= O(1/p_s^2)=O(R_s^2d_s^2)$.
By Chebyshev's inequality, we have that the returned estimate $\frac{1}{k} \sum_{i=1}^k S_i = (1 \pm \eps ) R_s d_s$ with constant probability.

\paragraph{Escape time algorithm}

In Step 1.~of the commute time algorithm, we run a random walk from $s$ until it either returns to $s$ or it hits $M$.
The statistic that we are interested in is the number of returns to $s$ before hitting $M$.
However, this is equivalent to the number of returns to $s$ before \emph{escaping} from $s$.
Hence, we can actually interrupt the walk after $O(\ET_s)$ steps instead of $O(\HT_s)$ steps, and this can lead to a faster algorithm.

Given an upper bound $\hat\ET_s \geq \ET_s$, the following algorithm requires $O(\hat\ET_s/\eps^3)$ steps of a random walk.

\begin{algorithm}[H]
\caption*{\bf Escape time algorithm for estimating $R_s d_s$} \label{alg:ET-R}
\begin{algorithmic}
\vspace{1mm}
\State
Let $\hat\ET_s \geq \ET_s$ be an upper bound on the escape time.\\
Repeat for $i = 1,2,\dots,k \in O(1/\eps^2)$:
\begin{enumerate}
\item
From source $s$, run $\hat\ET_s/\eps$ steps of a random walk (absorbing in $M$).
Let $\hat S_i$ be the number of visits to $s$ during the walk.
\end{enumerate}
Return $\frac{1}{k} \sum_{i=1}^k \hat S_i$.
\end{algorithmic}
\end{algorithm}

To check correctness, first we check that $\hat S_i$ is a good estimate for the actual number of returns $S_i$ before hitting $M$.
To this end, note that $\hat S_i \le S_i$ and 
\begin{align*}
\E&_s(S_i-\hat S_i) \\
&= \Pr \bigg( \esc_s < \frac{\hat\ET_s}{\eps} \bigg) \; \E\bigg(S_i-\hat S_i \;\Big|\; \esc_s < \frac{\hat\ET_s}{\eps} \bigg) + \Pr\bigg(\esc_s \ge \frac{\hat\ET_s}{\eps} \bigg) \; \E\bigg(S_i-\hat S_i \;\Big|\; \esc_s \ge \frac{\hat\ET_s}{\eps} \bigg) \\
&= 0 + \Pr\bigg(\esc_s \ge \frac{\hat\ET_s}{\eps} \bigg) \E\bigg(S_i-\hat S_i \;\Big|\; \esc_s \ge \frac{\hat\ET_s}{\eps} \bigg)
\le \eps \E(S_i).
\end{align*}
Hence, we have $(1-\eps)R_s d_s \le \E(\hat S_i) \le  R_s d_s$.
Furthermore $\hat S_i \le S_i$ implies that $\E(\hat S_i^2) \le \E(S_i^2) \in O(R_s^2 d_s^2)$.
By again applying Chebyshev's inequality, we get that the mean over $O(1/\eps^2)$ samples from $\hat S_i$ gives with constant probability an $\eps$-multiplicative approximation to $R_s d_s$, and hence $\frac{1}{k} \sum_{i=1}^k \hat S_i = (1 \pm \eps) R_s d_s$.

\subsubsection{Electric flow sampling with random walks}

A different, potential application of the escape time is related to sampling from the electric flow from $s$ to $M$, and thus generating a single step of the elfs process.
From the random walk coupling in \cref{sec:RW-coupling} it is clear that such an electric flow sample can be obtained from a random walk with a stopping rule whose expected length equals the escape time $\ET_s/2$.
However, the stopping rule in \cref{eq:stopping-prob} requires exact knowledge of the electric vertex potentials in the electric $s$-$M$ flow.
This seems hard to ensure in general, and hence the stopping rule as stated seems infeasible to actually implement.
Nevertheless, we conjecture that there does exist a random walk algorithm with a similar runtime, at least for \emph{approximately} sampling from the elfs process.

\begin{conjecture} \label{conj:RW-elfs}
For any $\eps > 0$, it is possible to sample $\eps$-approximately from the electric flow from $s$ to $M$ using $O(\ET_s \poly(1/\eps))$ random walk steps.
\end{conjecture}

We anticipate that a resolution of this conjecture would have various applications.
One example is an open question from \cite{andoni2018solving} about sublinear algorithms for sampling according to the $\ell_2$-weights of the solution of a linear system (similar to quantum algorithms for linear system solving \cite{harrow2009quantum}).
Since the electric flow can be interpreted as the solution of a linear system in the edge-vertex incidence matrix, a resolution of \cref{conj:RW-elfs} implies a sublinear algorithm for sampling from the $\ell_2$-weights of this linear system.

\newpage
\part{Trees} \label{part:trees}

In this part we study the electric hitting time on trees.
In \cref{sec:examples} we already discussed the $n$-vertex line with a sink on one end, whose electric hitting time is easily shown to be $O(\log n)$.
Extending this to a line with sinks on both ends is significantly more involved.
In this section, we show that in fact a similar, logarithmic bound on the electric hitting time holds for any trees.
This is captured by the following theorem, which is the main technical contribution of our work.

\begin{theorem}[Electric hitting time on trees]
\label{thm:elfstreebound}
Let $T$ be a tree, $s$ an initial source vertex and $M \subseteq V$ a sink.
For each $m \in M$, let $w_m$ be the weight of the edge incident to $m$, and let $f_m$ be the flow on this edge in the unit electric flow from $s$ to $M$.
Then it holds that
\[
\EHT_s
\leq 2 + \sum_{m \in M} f_m \log\left(\frac{R_s w_m}{f_m^2}\right).
\] 
\end{theorem}

We mention some instructive ways of rewriting this bound.
If the tree is unweighted (and so $w_m=1$ for all $m \in M$), then
\[
\EHT_s
\leq 2+\log R_s + 2\sum_{m\in M} f_m \log\left(\frac{1}{f_m}\right)
= 2 + \log R_s + 2 H(\mu),
\]
where $H(\mu)$ is the entropy of the arrival distribution $\mu$ at the sink vertices.
Since $R_s \leq n$ and $H(\mu) \leq \log|M|$, we have $\EHT_s = O(\log n)$.

In the case of a weighted tree, we can write
\begin{align*}
\EHT_s
&\leq 2 + \log (R_s d_M) + \sum_{m\in M} f_m \log\left(\frac{1}{f_m}\right) + \sum_{m\in M} f_m\log\left(\frac{w_m}{d_Mf_m}\right) \\
&= 2 + \log(R_s d_M) + H(\mu) - D(\mu || \sigma_M )
\end{align*}
where $d_M=\sum_{m \in M} w_m$ and $D(\mu || \sigma_M ) \geq 0$ is the relative entropy between the arrival distribution $\mu$ and $\sigma_M$, the stationary distribution restricted to $M$, defined by $\sigma_M(m) = w_m/d_M$.
Using the bound $R_s < n/w_{\min}$, this is $O(\log(n w_{\max}/w_{\min}))$, with $w_{\max}$ and $w_{\min}$ the maximum and minimum edge weights, respectively.

We prove Theorem~\ref{thm:elfstreebound} using an inductive argument.
First, in \cref{sec:schur}, we introduce the notion of Schur complement for graphs and random walks, and we show that it also applies to the elfs process on trees.
This allows us to contract a part of the tree without affecting the elfs process on the uncontracted part.
Then, in \cref{sec:EHT-trees}, we use this to prove an inductive step that relates the number of steps of the elfs process in a subtree to that in smaller subtrees of that subtree.
The base case of the inductive argument corresponds to edges directly incident on the sink, which we can treat explicitly.

\section{Schur complement and cut vertices} \label{sec:schur}

The Schur complement of a graph with respect to some subset $S$ is a useful tool.
Intuitively, it marginalizes out the vertices outside $S$, in such a way that the behavior of a random walk on $S$ is still recovered.
We will use this to achieve a divide-and-conquer approach for bounding the electric hitting time on trees.

\subsection{Schur complement of matrices and graphs} \label{sec:Schur}

We first define the Schur complement for general matrices, and then zoom in on the case of Laplacian matrices and graphs.
We refer the interested reader to \cite{zhang2006schur} for the matrix case, and to \cite{kyng2017approximate} for the Laplacian case.

For a general square matrix $M$, consider a partition $(S,C)$ of its indices so that $M=\begin{pmatrix} M_{SS}& M_{SC} \\ M_{CS} & M_{CC} \end{pmatrix}$. If $M_{CC}$ is invertible, then the \emph{Schur complement}  $M^c$ of $M$ on $S$ as:
 \begin{equation}M^c=M_{SS}- M_{SC}M_{CC}^{-1}M_{CS}.
 \label{eq:Schurcomplement}
 \end{equation}
This definition ensures the following key fact about the block inverse of $M$:\footnote{If $M$ is Laplacian then we can replace the inverse with the pseudoinverse \cite[Fact 2.3.2]{kyng2017approximate}.}
 \begin{equation}(M^c)^{-1}=(M^{-1})_{SS}.
 \label{eq:Schurinverse}\end{equation}
In words, the inverse of $M$ restricted to $S$ equals the inverse of the Schur complemented matrix $M^c$.

We are interested in the case where $M = L$, with $L$ the Laplacian of some graph $G$.
We assume that $G$ is connected, which ensures that $M_{CC}$ is invertible.
In such case, it turns out that the Schur complement $L^c$ of $L$ on $S$ is again the Laplacian of a graph $G'$ with vertex set $S$.
We call this the Schur complemented graph.
From the definition of the Schur complement, we can see that the weight of an edge $(x,y)$, for $x, y \in S$, only changes if $x$ and $y$ are both in the boundary with $C$ (i.e., both $x$ and $y$ have a neighbor in $C$).\footnote{To be very precise, the Schur complement also introduces selfloops on the boundary vertices of $S$. However, these do not influence the elfs process and so we can safely ignore these.}

The property in \cref{eq:Schurinverse} plays an important role when considering electric flows on the graph.
Consider an external net current or demand $i_{ext} \in \mathbb{R}^n$ (e.g., in a unit $s$-$t$ flow we have $i_{ext} = e_s - e_t$).
Then the electric flow corresponding to this external current has voltages $v = L^+ i_{ext}$ \cite{kemeny1983finite}.
Now if $C$ contains no edges where flow enters or leaves (i.e., $i_{ext}(v) = 0$ for all $v \in C$), then by \cref{eq:Schurinverse} the voltages in $S$ are the same when routing $i_{ext}$ in $G$ or the Schur complemented graph $G'$.

Another important fact is that a random walk on $G'$ behaves the same as a random walk on $G$ restricted to vertices in $S$ (i.e., ignoring all vertices of the random walk outside $S$).
To see this, let $P=I-D^{-1}L$ be the probability transition matrix  of the random walk in $G$ and let $P'$ be the probability transition matrix of the walk when ignoring vertices in $C$. 
In the walk on $G$ from a vertex in $S$, the next vertex in $S$ can be reached either by going straight there, or by stepping to a vertex in $C$ and back to $S$, or by stepping to $C$ and staying in $C$ for one step and then back to $S$, and so on. 
Therefore
\begin{align*}
P' & = P_{SS}+P_{SC}(I+P_{CC}+P_{CC}^2+\dots)P_{CS} \\
& =P_{SS}+P_{SC}(I-P_{CC})^{-1}P_{CS} \\
& = I - D_{SS}^{-1} (L_{SS}-L_{SC}L_{CC}^{-1}L_{CS})
\end{align*} 
which is the transition matrix for the random walk on the Schur complemented graph $G'$ with Laplacian $L^c=L_{SS}-L_{SC}L_{CC}^{-1}L_{CS}$.

\subsection{Elfs process and cut vertices}
In the previous section, we saw how the Schur complement makes it easy to analyse how the random walk on the whole graph $G$ behaves when restricted to a subset of vertices $S$. We would like to prove a similar result to hold for the elfs process, but are only able to do so in the special case where there is a single boundary or \emph{cut vertex}.

A cut vertex, if removed, increases the number of connected components of the graph.
On the left of \cref{fig:schur} we depict the generic case of a cut vertex $s$.
The graph $G$ has vertex set $V_A \cup \{s\} \cup V_B$ and edge set $A \cup B$ where $A$ consists of edges of the form $(x,y)$ for $x,y  \in V_A\cup \{s\}$, and $B$ contains edges of the form $(x,y)$ for $x,y  \in V_B\cup \{s\}$.

\begin{figure}[htb]
 \centering
 \begin{subfigure}[t]{0.45\textwidth}
 \centering
 \begin{tikzpicture}
 \tikzstyle vertex=[draw=black,circle]

 \node[vertex,label=below:s] (s) at (2,0) {};

\draw[dashed] (1,0) ellipse (1 and 0.5) ;
\node at (1,0) {$A$};

\draw[dashed] (3,0) ellipse (1 and 0.5) ;
\node at (3,0) {$B$};
\end{tikzpicture}
\label{fig:cutvertex}
\end{subfigure}
\hfill
\begin{subfigure}[t]{0.45\textwidth}
\centering
\begin{tikzpicture}
\tikzstyle vertex=[draw=black,circle]
 \node[vertex,label=below:s] (a) at (2,0) {};
 \node[vertex,label=below:b] (b) at (4,0) {};
 \draw (a) -- (b);

\draw[dashed] (1,0) ellipse (1 and 0.5) ;
\node at (1,0) {$A$}; 
 \end{tikzpicture}
\label{fig:Schurcutvertex}
 \end{subfigure}
\caption{(left) A graph $G$ with cut vertex $s$. (right) The Schur complement of $G$ with vertices in $V_B\backslash b$ removed.}
\label{fig:schur}
 \end{figure}
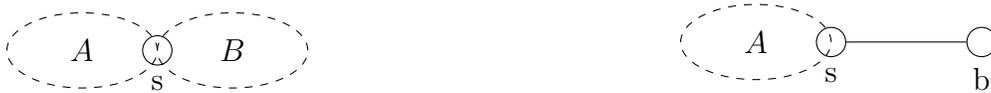

We will consider random walks and elfs processes on $G$ with a sink $M \subseteq V$ of absorbing vertices.
We can merge the sink vertices in $M\cap V_B$ to a single vertex $b$, as this will not affect the random walk or the elfs process.
Now consider the Schur complement of $G$ on $V_A \cup \{s,b\}$, which removes all vertices in $V_B\backslash b$, and let $G'$ denote the resulting graph (depicted on the right in \cref{fig:schur}).  Since the only vertices in the boundary of $V_B\backslash b$ are $b$ and $s$, there is only one new edge $(s,b)$ in $G'$, and the rest of the edges in $A$ remain the same as in $G$.\footnote{Recall that we can ignore selfloops introduced by the Schur complement.}
By the definition of the Schur complement, the weight of the new edge $(s,b)$ equals $w_{sb} = R_{s,b}^{-1}$, with $R_{s,b}$ the effective resistance between $s$ and $b$ in the original graph $G$.

We are now ready to state how the local behaviour of the elfs process in $A$ is the same for both $G$ and $G'$.
In slight contrast with the other parts of this work, we here focus on the \emph{edges} sampled in the elfs process rather than the vertices.
\begin{lemma}[Elfs schur complement]
\label{lem:ElfsSchur}
Let $G$ be a graph with cut vertex $s$ (left of \cref{fig:schur}) and sink $M$, with $M \cap V_B = \{b\}$.
Let $G'$ be the Schur complemented graph on $V_A \cup \{s,b\}$. We consider the elfs process starting at a vertex in $V_A \cup \{s\}$.
If $X=\{X_1,X_2,\dots\}$ is the sequence of edges sampled in $A$ during the elfs process on $G$, and $Y=\{Y_1,Y_2,\dots\}$ is the sequence of edges sampled in $A$ during the elfs process on $G'$, then $X$ and $Y$ have the same distribution.
\end{lemma}

To prove Lemma~\ref{lem:ElfsSchur}, let $ \tilde X=\{\tilde X_1, \tilde X_2,\dots\}$ and $ \tilde Y=\{\tilde Y_1, \tilde Y_2,\dots\}$  denote the sequence of all edges sampled in the elfs process on $G$ and $G'$, respectively.
Then $X \subseteq \tilde X$ and $Y \subseteq \tilde Y$ are the subsequences of edges that are in $A$.

As explained in Section~\ref{sec:Schur}, the electric flow from a vertex in $V_A\cup\{s\}$ to the sink $M$ will have the same voltages in both $G$ and $G'$.
This implies that the flows on edges in $A$, and the effective resistance between pairs of vertices in $A$, are equal in $G$ and $G'$.
As a consequence, the first edge sampled by the elfs process will be equally distributed: for any $e \in A$,
\[
\Pr(\tilde X_1=e)
= \Pr(\tilde Y_1=e)
\;\text{ and }\;
\Pr(\tilde X_1 \in B)=\Pr(\tilde Y_1=(s,b)).
\]
Similarly, for any $e,e' \in A$ it holds that
\[
\Pr(\tilde X_i=e \mid \tilde X_{i-1}=e')
= \Pr(\tilde Y_i=e \mid \tilde Y_{i-1}=e')
\]
and
\[
\Pr(\tilde X_i \in B \mid \tilde X_{i-1}=e')
= \Pr(\tilde Y_i=(s,b) \mid \tilde Y_{i-1}=e').
\]

Now consider the case where $\tilde X_i \in B$ or $\tilde Y_i = (s,b)$. If we condition on the next edge $\tilde X_{i+1}$ or $\tilde Y_{i+1}$ being in $A$, then in both cases it will be effectively sampled from the electric flow from $s$ to $M \cap A$, which is the same in both $G$ and $G'$.
It remains to check that this conditioning does not change the distribution of $\tilde X_{i+1}$ or $\tilde Y_{i+1}$.
To this end, we introduce a new quantity $p_{B \rightarrow A}$ that equals the probability that, after visiting an edge in $B$, the elfs process will at some point return to an edge in $A$.
We can write this as
\begin{align*}
p_{B \rightarrow A}
\coloneqq& \, \Pr(\exists t'>t \text{ such that } Y_{t'}\in A \mid Y_t \in B) \\
=& \, \Pr(\exists t'>1 \text{ such that } Y_{t'}\in A \mid Y_1 \in B).
\end{align*}
The equality follows from the fact that the quantity $p_{B \rightarrow A}$ is independent of $t$.
Indeed, let $t_0$ be the first time such that $Y_{t_0} \in B$.
Since there is a cut vertex that disconnect $B$ from the rest of the graph, the electric flow into $B$ (and hence the distribution of $Y_{t_0}$) from a source outside $B$ is independent of the source.
Hence, $Y_{t_0}$ is distributed according to a sample from the electric flow from $s$ to $M \cap B$, and so independent of any source or time $t_0$.

We have to show that $p_{B \rightarrow A}$ is the same for both $G$ and $G'$. 
We do this in the following lemma.
Specifically, we show that $p_{B\rightarrow A}$ depends only on $f_A=\sum_{m \in M\cap A} f_m$, which is the total flow into $A$ from the unit electric $s$-$M$ flow.
As this is the same for both $G$ and $G'$, this implies the claimed statement, and hence proves \cref{lem:ElfsSchur}.

 \begin{lemma} \label{lem:pBA}
 Consider the elfs process from source $s$ to sink $M$, with $s$ a cut vertex that splits the graph into two components with edge sets $A$ and $B$. Let $f_A$ be the total flow into $A$, and let $Y_1,Y_2,\dots Y_{\rho}$ be the edges sampled in the elfs process.
 Then
 \begin{enumerate}
 \item $\Pr(Y_1 \in A)=f_A$,
 \item $p_{B \rightarrow A} =\dfrac{f_A}{1+f_A}$.
 \end{enumerate}
 \end{lemma}
 
 \begin{figure}[H]
 \centering
 \begin{tikzpicture}
 \tikzstyle vertex=[draw=black,circle]

 \node[vertex,label=below:s] (s) at (2,0) {};

\draw[dashed] (1,0) ellipse (1 and 0.5) ;
\node at (1,0) {$A$};

\draw[dashed] (3,0) ellipse (1 and 0.5) ;
\node at (3,0) {$B$};
 \end{tikzpicture}
 
 \end{figure}
 
 \begin{proof}
 We start by proving item 1.
 Let $G'$ denote the Schur complement of $G$ with all edges in $B$ contracted to a single edge $(s,b)$. Since the edge weights in $A$ are unchanged and the electric flow (and effective resistance) from $s$ to $M$ in $A$ is the same in both $G$ and $G'$, the probability $\Pr(Y_1 \in A)$ is also unchanged.
 Let $f_{sb}=1-f_A$ be the flow down edge $(s,b)$. Then 
 \[\Pr(Y_1 \notin A)=\Pr(Y_1=(s,b)) = \frac{f_{sb}^2}{R_s}w_{sb}=\frac{f_{sb}(v_s-v_b)}{R_s}=f_{sb},\]
 using that $v_s = R_s$ and $v_b = 0$.
 As $f_{sb} = 1 - f_A$, this implies that $\Pr(Y_1 \in A)=f_A$.
 
 We prove item 2. in three stages: (i) first we consider the special case where $B$ is a single (possibly weighted) edge $(s,b)$; (ii) then we consider the case where $A$ is a single (possibly weighted) edge $(a,s)$; (iii) finally we consider the general case.
 
\emph{Stage (i).}
Assume that the first sampled edge $Y_1 \in B$, i.e. $Y_1= (s,b)$. Then with probability $1/2$, $b$ is chosen and the process terminates. Otherwise, with probability 1/2, $s$ is chosen and a new sample from the electric flow from $s$ is taken, which either selects an edge in $A$ with probability $f_A$, or it selects the edge $(s,b)$ with probability $1-f_A$.
If the edge $(s,b)$ is sampled again then we are back to the initial situation.
Therefore 
\[p_{B\rightarrow A} = \frac{1}{2} (f_A+(1-f_A)p_{B\rightarrow A})\]
which can be rearranged to $p_{B\rightarrow A} = \frac{f_A}{1+f_A}$. 
 
\emph{Stage (ii).}
Now we consider the case where $A$ is a single edge, but $B$ is arbitrary. We will express $\Pr(Y_{\rho}\in A)$, which is the probability of being absorbed in $M \cup A$, in terms of $p_{B \rightarrow A}$ and other known parameters. Since we know  $\Pr(Y_{\rho}\in A)=f_A$, we can then find $p_{B \rightarrow A}$:
\begin{align}
\Pr(Y_{\rho}\in A) &= \Pr(Y_1 \in A) \Pr(Y_{\rho} \in A | Y_1 \in A) + \Pr(Y_1 \in B) \Pr(Y_{\rho} \in A | Y_1 \in B) \nonumber \\
&= f_A \Pr(Y_{\rho} \in A | Y_1 \in A) + (1-f_A) p_{B\rightarrow A} \Pr(Y_\rho \in A | Y_1 \in A) \nonumber \\
&=\left[ f_A+(1-f_A)p_{B \rightarrow A}\right] \Pr(Y_\rho \in A | Y_1 \in A)
\label{eq:Pr_endinA}
\end{align}
where we have substituted in $\Pr(Y_1 \in A)=f_A$, $\Pr(Y_1 \in B) = 1-f_A$ and $\Pr(Y_{\rho} \in A | Y_1 \in B) = p_{B\rightarrow A} \Pr(Y_\rho \in A | Y_1 \in A)$.

Now $\Pr(Y_\rho \in A |Y_1 \in A)$, the probability that the elfs process ends in $A$ given that the first edge sampled is in $A$, is equal to the probability that all samples are from $A$, plus the probability that at some point we switch to $B$ and back and then stay in $A$, plus the probability that we switch to $B$ and back twice and then stay in $A$, and so on. Therefore 
\begin{align*}
\Pr(Y_{\rho}\in A |  Y_1 \in A)
&= \Pr(Y_i \in A \: \forall i | Y_1 \in A)\sum_{k=0}^\infty(p_{B\rightarrow A}p_{A \rightarrow B})^k \\
&= (1-p_{A\rightarrow B})\frac{1}{1-p_{B \rightarrow A}p_{A \rightarrow B}}
\end{align*}
where we have used the geometric sum and $\Pr(Y_i \in A \: \forall i | Y_1 \in A)= 1-p_{A\rightarrow B}$ since this is the probability of staying in $A$ forever.

Substituting into \cref{eq:Pr_endinA} gives 
\[f_A= \left[ f_A+(1-f_A)p_{B \rightarrow A}\right] (1-p_{A\rightarrow B})\frac{1}{1-p_{B \rightarrow A}p_{A \rightarrow B}} \]
Now we use the fact that $A$ is a single edge, and so by (i),  $p_{A \rightarrow B} =\frac{f_B}{1+f_B}=\frac{1-f_A}{2-f_A}$. Substituting this in leads to the claimed result.

\emph{Stage (iii).}
Stages (i) and (ii) imply the result for the general case, as we can  take the Schur complement on $A$ so that it is a single edge.
This leaves the quantity $p_{B \rightarrow A}$ invariant.
 \end{proof}

\section{Electric hitting time on trees} \label{sec:EHT-trees}
We use the ideas of the previous section, and \cref{lem:ElfsSchur} in particular, to prove Theorem~\ref{thm:elfstreebound} and bound the electric hitting time $\EHT_s$ on trees.
On a tree, every vertex is a cut vertex, and this makes Lemma \ref{lem:ElfsSchur} particularly easy to use.

We actually prove a slightly stronger result in Lemma~\ref{lem:elfstreebound} below.
Before we state the lemma, we need to introduce some notation.
We consider a tree $T$ with a fixed sink $M$.
For a vertex $x$ and subset of edges $S$ of $T$, we let $E_x(S)$ denote the expected number of edges sampled from $S$ in the elfs process starting at vertex $x$.
We also use the shorthand $f_{x \rightarrow y}$ for $f_{p(y),y}$ with $p(y)$ the parent of $y$.
I.e., $f_{x \rightarrow y}$ denotes the flow from $p(y)$ to $y$ in the unit electric flow from $x$ to $M$.

\begin{lemma}
\label{lem:elfstreebound}
Fix a root $r$ of a tree $T$.
For a vertex $x$, let $T_x$ be the subtree down from $x$, and let $M_x = M \cap T_x$ be the set of sink vertices below $x$. 
Then
\[E_x(T_x)\le\sum_{m \in M_x} f_{x \rightarrow m} \left( 2+\log \left(\frac{R_x w_m}{f_{x \rightarrow m}^2}\right)\right) \]
\end{lemma}

Theorem~\ref{thm:elfstreebound} immediately follows by setting $x = r$ and $r = s$.
To prove Lemma~\ref{lem:elfstreebound}, we first show how the expected time in a subtree $E_x(T_x)$ varies as the starting vertex $x$ changes (\cref{sec:recurrence}).
Then we use a proof by induction on the depth of $T_x$ (\cref{sec:lem7}).

\subsection{Recurrence relation} \label{sec:recurrence}

Consider an edge $e_2$ in the tree $T$ with endpoints $a,b \notin M$. Let the subtree $T_1$ (respectively $T_3$) be the rest of the tree $T$ connected to $a$ (respectively $b$). So $T$ looks like this:

 \begin{figure}[H]
 \centering
 \begin{tikzpicture}
 \tikzstyle vertex=[draw=black,circle]

 \node[vertex,label=below:a] (a) at (2,0) {};
 \node[vertex,label=below:b] (b) at (4,0) {};

\draw[dashed] (1,0) ellipse (1 and 0.5) ;
\node at (1,0) {$T_1$};
\draw[dashed] (5,0) ellipse (1 and 0.5) ;
\node at (5,0) {$T_3$};
 \draw (a) to node[above] {$e_2$} (b);
 
 \end{tikzpicture}
 
 \end{figure}

Let $p=f^{a \rightarrow M}_{ab}$ be the flow from $a$ to $b$ in the electric flow from $a$ to $M$ and $q=f^{b\rightarrow M}_{ba}$ be the flow from $b$ to $a$ in the electric flow from $b$ to $M$.
Then the following lemma holds.

\begin{lemma}
\label{lem:treerecurrence}
\begin{enumerate}
    \item $E_b(T_1)=\frac{q}{1-p}E_a(T_1)$ 
    \item $E_b(e_2)=q\frac{1-p-q}{1-p-q+2pq} \le q \log \left(\frac{(1-p)(1-q)}{pq}\right)$.
\end{enumerate}
\end{lemma}

It will prove useful to consider the Schur complement of $T$ on $\{a,b,M\}$.
Again, without affecting the dynamics, we can merge the sinks in $T_1$ to a single sink $m_a$, and those in $T_2$ to a single sink $m_b$, to obtain a Schur complemented graph $T^c$ that looks like this:
 \begin{figure}[h!]
 \centering
 \begin{tikzpicture}
 \tikzstyle vertex=[draw=black,circle]

 \node[vertex, label=below:$m_a$] (left) at (0,0) {};
 \node[vertex, label=below:$a$] (a) at (2,0) {};
 \node[vertex, label=below:$b$] (b) at (4,0) {};
 \node[vertex, label=below:$m_b$] (right) at (6,0) {};

 \draw (left) to node[above] {$e_1$} (a);
 \draw (a) to node[above] {$e_2$} (b);
 \draw (b) to node[above] {$e_3$} (right);
 
 \end{tikzpicture}
 \end{figure}
 
\noindent
Let $r_i$ denote the resistance of edge $e_i$ in $T^c$.
Then we can calculate the effective resistances $R_a$ (resp.~$R_b)$ between $a$ (resp.~$b$) and the sinks as
\[
R_a
= \frac{r_1(r_2+r_3)}{r_1+r_2+r_3}
\quad \text{ and } \quad
R_b
= \frac{r_3(r_1+r_2)}{r_1+r_2+r_3}.
\]
From this we can then calculate the flows $p=f^{a \rightarrow M}_{ab}$ and $q=f^{b\rightarrow M}_{ba}$ as 
\[
p
= \frac{R_a}{r_2+r_3}=\frac{r_1}{r_1+r_2+r_3},
\quad \text{ and } \quad
q
= \frac{R_b}{r_1+r_2}=\frac{r_3}{r_1+r_2+r_3}.
\]
 
As a final quantity of interest, let $x$ be a vertex and $S$ a subset of the edges of $T$, and let $P_x(S)$ be the probability that a sample of the flow from $x$ to $M$ is in $S$.
Then we can calculate
\begin{equation}
\label{eq:exps}
\begin{alignedat}{3}
  &P_a(T_1) = \frac{(1-p)^2 r_1}{R_a}=1-p
  \qquad\qquad\qquad
  &&P_b(T_1) = \frac{q^2 r_1}{R_a}=\frac{qp}{1-q} \\
  &P_a(e_2) = \frac{p^2r_2}{R_a}=\frac{p(1-p-q)}{1-p} 
  &&P_b(e_2) = \frac{q^2r_2}{R_b}=\frac{q(1-p-q)}{1-q}\\
  &P_a(T_3) = \frac{p^2r_3}{R_a}=\frac{pq}{1-p}
  &&P_b(T_3) = \frac{(1-q)^2r_3}{R_b}=1-q.
\end{alignedat}
\end{equation}
We now turn to the proof of \cref{lem:treerecurrence}.

\begin{proof}[{\cref{lem:treerecurrence}, item 1.}]
By Lemma~\ref{lem:ElfsSchur} we can replace the edges of $T_3$ with a single edge $e_3$:

 \begin{figure}[H]
 \centering
 \begin{tikzpicture}
 \tikzstyle vertex=[draw=black,circle]

 \node[vertex,label=below:a] (a) at (2,0) {};
 \node[vertex,label=below:b] (b) at (4,0) {};
 \node[vertex] (right) at (6,0) {};

\draw[dashed] (1,0) ellipse (1 and 0.5) ;
\node at (1,0) {$T_1$};

 \draw (a) to node[above] {$e_2$} (b);
 \draw (b) to node[above] {$e_3$} (right);

 \end{tikzpicture}
 
 \end{figure}

 Let $E_{T_1}(T_1)$ be the expected number of visits of the elfs process to $T_1$ conditioned on the first edge sample being in $T_1$ (note that this is independent of the initial source vertex).
 By considering the three possibilities of the first sampled edge from $b$ (i.e., $T_1$, $e_2$ or $e_3$), we get the following expression for $E_b(T_1)$:
 \begin{equation}
 E_b(T_1)
 =  P_b(T_1)E_{a|T_1}(T_1) + P_b(e_2)\frac{E_a(T_1)+E_b(T_1)}{2}
    + P_b(e_3)\frac{E_b(T_1)}{2}. 
 \label{eq:E_b(T_1)}
 \end{equation}
 
We first re-express $E_{T_1}(T_1)$ in terms of $p$ and $E_a(T_1)$:
\[
E_a(T_1)
= \Pr_a(\text{visit }T_1 \text{ at some point}) \, E_{T_1}(T_1)
= \frac{2(1-p)}{2-p} \, E_{T_1}(T_1).
\]
The last inequality follows from
\begin{align*}
\Pr_a(\text{visit }T_1 \text{ at some point})
&= \Pr_a(Y_1 \in A) + p_{B \rightarrow A} \Pr_a(Y_1 \in B) \\
&= 1-p + \frac{1-p}{2-p} p
= \frac{2(1-p)}{2-p},
\end{align*}
using that $p=f^{a \rightarrow M}_{ab} = \Pr_a(Y_1 \in B)$ and $p_{B \rightarrow A} = f^{a \rightarrow M}_{ab}/(1+f^{a \rightarrow M}_{ab}) = (1-p)/(2-p)$ (\cref{lem:pBA}).

Substituting the expression for $E_a(T_1)$ into \cref{eq:E_b(T_1)}, and rearranging, gives
\[\left(1-\frac{P_b(e_2)}{2} -\frac{P_b(e_3)}{2}\right) E_b(T_1) = \left( P_b(T_1) \frac{1+f}{2f}+P_b(e_2)\right)E_a(T_1)\]
\[ \Leftrightarrow \qquad \frac{1-q(1-p)}{2(1-q)}E_b(T_1)=\frac{q}{1-p}\frac{1-q(1-p)}{2(1-q)}E_a(T_1)\]

Since $0<p,q<1$, we can divide through by $\frac{1-q(1-p)}{2(1-q)}$ to get the claimed result.
\end{proof}

\begin{proof}[{\cref{lem:treerecurrence}, item 2.}]
By Lemma~\ref{lem:ElfsSchur} we can consider the aforementioned Schur complement of $T$ on $\{a,b,m_a,m_b\}$, with sink $M = \{m_a,m_b\}$.
The resistance of the new edge $e_1$ (resp.~$e_2$) is the effective resistance in $T$ between $a$ and $M\cap T_1$ (resp.~between $b$ and $M \cap T_2$).
  
  \begin{figure}[h!]
 \centering
 \begin{tikzpicture}
 \tikzstyle vertex=[draw=black,circle]

 \node[vertex, label=below:$m_a$] (left) at (0,0) {};
 \node[vertex, label=below:$a$] (a) at (2,0) {};
 \node[vertex, label=below:$b$] (b) at (4,0) {};
 \node[vertex, label=below:$m_b$] (right) at (6,0) {};

 \draw (left) to node[above] {$e_1$} (a);
 \draw (a) to node[above] {$e_2$} (b);
 \draw (b) to node[above] {$e_3$} (right);
 
 \end{tikzpicture}
 \end{figure}
 \noindent
 
For the elfs process from $b$ to $M$, we again distinguish the three possibilities of sampling the first edge ($e_1$, $e_2$ or $e_3$).
This leads to the following expression for $E_b(e_2)$:
\[
E_b(e_2)
= P_b(e_1) \frac{E_a(e_2)}{2}
 + P_b(e_2) \left(1 + \frac{E_a(e_2)+E_b(e_2)}{2}\right) + P_b(e_3)\frac{E_b(e_2)}{2}.
\]
Similarly, when starting from $a$ we have:
\[
E_a(e_2)
= P_a(e_1)\frac{E_a(e_2)}{2}+P_a(e_2)\left(1+\frac{E_a(e_2)+E_b(e_2)}{2}\right) +P_a(e_3)\frac{E_b(e_2)}{2}.
\]

This yields a linear system in $E_a(e_2)$ and $E_b(e_2)$.
Solving the system, and substituting the expressions in \cref{eq:exps}, we get
\[
E_a(e_2)
= p\frac{1-p-q}{1-p-q+2pq}
\quad\text{ and }\quad
E_b(e_2)
= q\frac{1-p-q}{1-p-q+2pq},
\]
which proves the equality of item 2.
 
For the inequality, observe that 
\[
0
\le\frac{1-p-q}{1-p-q+2pq}
\le 1.
\]
This holds because (i) there is more flow to $m_a$ in the electric flow from $a$ than in the electric flow from $b$, and so $1-p\ge q$ (or, equivalently, $1-p-q \ge 0$), and (ii) both $p$ and $q$ cannot be negative so $2pq\ge 0$.
We can therefore use the inequality\footnote{This holds for base 2 logarithm.} $x \le \log(x+1)$ which holds for $0\le x \le 1$ to get 
\begin{equation}\frac{1-p-q}{1-p-q+2pq} \le \log \left(\frac{2(1-p)(1-q)}{1-p-q+2pq}\right).
\end{equation}
Now we use the fact that $ 1-p-q \ge 0$ again in the denominator to get 
\begin{equation}\frac{2(1-p)(1-q)}{1-p-q+2pq} \le  \frac{(1-p)(1-q)}{pq}.
\end{equation}
Combining these two equations and using the monotonicity of the logarithm function gives the desired inequality.
\end{proof}

\subsection{Subtree lemma} \label{sec:lem7}
 
We can now finalize the proof of \cref{lem:elfstreebound}.

\begin{proof}[Proof of \cref{lem:elfstreebound}]
The proof works by induction on the depth of $T_x$. 

If $T_x$ has depth 1, then after each visit to an edge in $T_x$, the next vertex is either a sink vertex with probability $\frac{1}{2}$ or it returns to $x$. The probability of ending at a sink vertex in $M_x$ is $\sum_{m \in M_x} f_{x \rightarrow m}$, and so the expected number of visits is $2\sum_{m \in M_x} f_{x \rightarrow m}$.
Furthermore, $f_{x \rightarrow m}=R_x w_m$, so $\log(R_x w_m /f_{x \rightarrow m}^2)=\log(1/f_{x \rightarrow m}) \ge 0$.

Now assume by induction that the lemma holds for all children $c$ of a vertex $b$. Denote the set of children of $b$ with $C(b)$. Then by Lemma~\ref{lem:treerecurrence}, 
 \begin{equation}E_b(T_b)= \sum_{c \in C(b)} E_b(bc)+E_b(T_c)
 \le \sum_{c \in C(b)} q_c\log\left(\frac{(1-p_c)(1-q_c)}{p_c q_c}\right) + \frac{q_c}{1-p_c}E_c(T_c)
 \label{eq:Ebinduction}
 \end{equation}
 where $p_c$ is the flow into $b$ in the electric flow starting at $c$, and $q_c$ is the flow into $c$ in the electric flow starting at $b$.

Now let $M_c = M \cap T_c$.
Note that $1-p_c = \sum_{m \in M_c} f_{c \rightarrow m}$, since this is the total flow to elements in $M_c$ in the flow starting at $c$. 
Multiplying the first term of the above \cref{eq:Ebinduction} by $1=\sum_{m \in M_c} f_{c \rightarrow m}/(1-p_c)$, and substituting in $E_c(T_c) \leq \sum_{m \in M_c}f_{c \rightarrow m}\left(2+\log(R_c/f_{c\rightarrow m}^2)\right)$ for the second term by the inductive hypothesis, we get:
 
\begin{align*}
E_b(T_b) & \le \sum_{c \in C(b)} \frac{q_c}{1-p_c}\sum_{m \in M_c} f_{c \rightarrow m} \left( \log\left(\frac{(1-p_c)(1-q_c)}{p_c q_c}\right) + 2+ \log\left(\frac{R_c w_m}{f_{c\rightarrow m}^2}\right) \right)
\nonumber \\
& =\sum_{c \in C(b)} \sum_{m \in M_c} f_{b \rightarrow m} \left(2+\log \left( \frac{q_c(1-q_c)}{p_c (1-p_c)} \frac{R_c w_m}{f_{b\rightarrow m}^2}\right)\right)  \\
& =\sum_{m \in M_b} f_{b \rightarrow m} \left( 2+\log \left(\frac{R_b w_m}{f_{b \rightarrow m}^2}\right)\right)
\end{align*}
where in the first equality we substituted in $\frac{f_{c \rightarrow m}}{f_{b \rightarrow m}}=\frac{1-p_c}{q_c}$, and in the final equality we susbstituted $\frac{R_c}{R_b}= \frac{p_c(1-p_c)}{q_c(1-q_c)}$.
\end{proof}

\newpage
\part{Quantum walks}

\section{Preliminaries: linear algebra and electric networks}

The previous part of this work heavily relied on the connection between electric networks and random walks.
To establish the connection between electric networks and \emph{quantum} walks, we first recollect the useful connection between electric networks and linear algebra.
A more elaborate exposition can be found in \cite{lyons2017probability}.

To a given graph $G = (V,E,w)$ we associate two vector spaces called the node space and the edge space, defined respectively by\footnote{Throughout this part we use Dirac notation: $\ket{\psi}$ is a vector, $\bra{\psi}$ is its conjugate transpose, and $\braket{\phi|\psi}$ is the inner product between $\ket{\phi}$ and $\ket{\psi}$.}
\[
\ell(V) = \linspan\{\ket{x} \mid x \in V\}
\quad \text{ and } \quad
\ell(E) = \linspan\{\ket{xy},\ket{yx} \mid (x,y) \in E\}.
\]
The (anti-)symmetric subspaces $\ell^{\pm}(E) \subset \ell(E)$ are spanned by the vectors $\ket{\psi} \in \ell(E)$ such that $\braket{xy|\psi} = \pm \braket{yx|\psi}$.
The (normalized) \textit{star state} $\ket{\phi_x}$ at $x$ is defined as
\[
\ket{\phi_x}
= \frac{1}{\sqrt{d_x}} \sum_y \sqrt{w_{xy}} \ket{xy}.
\]
The span of these star states defines the \emph{star subspace} $\ell^*(E) \subset \ell(E)$.
To an (antisymmetric) flow $h:V \times V \to \R$ on $G$ we associate a (normalized) flow vector $\ket{h} \in \ell^-(E)$, defined by
\[
\ket{h}
= \frac{1}{\sqrt{2 \cE(h)}}
    \sum_{xy} \frac{h_{xy}}{\sqrt{w_{xy}}} \ket{xy},
\]
where we recall the energy of flow $h$ as $\cE(h) = \sum_{x,y} h_{xy}^2/(2w_{xy})$.\footnote{The additional factor $1/2$ corrects for double counting the edges.}
We say that a flow $h$ is \textit{closed outside $M$} if the net flow out of any vertex $x \notin M$ is zero, or equivalently
\[
\braket{\phi_x|h}
\propto \sum_y h_{xy}
= 0, \quad \forall x \notin M.
\]
Now let $\ell^-_\cl(E)$ denote the subspace spanned by all closed flows outside $M$, and $\Pi_\cl$ the orthogonal projector onto $\ell_\cl^-(E)$.
For source node $s \in V$ and sink $M \subseteq V$, an \textit{$s$-$M$ flow} $h$ yields a vector $\ket{h} \in \ell^-(E)$ for which
\[
\braket{\phi_s|h} = - \sum_{y \in M} \braket{\phi_y|h} \neq 0
\quad \text{ and } \quad
\braket{\phi_x|h} = 0,\; \forall x \notin \{s,M\}.
\]
In particular, an $s$-$M$ flow is closed outside $\{s,M\}$.
The \emph{electric} flow $f$ from $s$ to $M$ is the unique unit $s$-$M$ flow with minimal energy $\cE(f) = R_s$.
Alternatively, it is the unique unit $s$-$M$ flow that obeys \emph{Kirchhoff's cycle law} \cite{lyons2017probability}, which states that $\braket{g|f} = 0$ for any flow $g$ that is closed outside $M$.

\section{Quantum walks and electric networks}

We can define quantum walks using much of the terminology introduced earlier.
While a random walk is defined over the vertex space, a quantum walk takes place over the edge space.
I.e., the state of a quantum walk is described by a quantum state $\ket{\psi} \in \ell(E)$.
It can be helpful to think about the basis state $\ket{\psi} = \ket{x,y}$ as a quantum walk positioned on vertex $x$, but keeping a memory of the previous position $y$. 

To describe the dynamics, we first define the unitary operators
\[
D_x = 2 \ket{\phi_x} \bra{\phi_x} - I 
\quad \text{ and } \quad \bigoplus_x D_x.
\]
These describe reflections around the star state $\ket{\phi_x}$ and the star subspace $\ell^*(E)$, respectively.
We also define the $\swap$ operator by $\swap \ket{xy} = \ket{yx}$.
This operator describes a reflection around the symmetric subspace $\ell^+(E)$.
Following \cite{szegedy2004quantum,magniez2011search}, the quantum walk operator is then defined as the product of reflections $\swap \cdot \bigoplus_x D_x$.
This structure is close to the structure of a random walk on a graph, where the reflection $\bigoplus_x D_x$ mixes up the second (memory) register of the walk, while the $\swap$ operator corresponds to an actual move along the resulting edge.

In this work, we use a slight variation called the \emph{absorbing} quantum walk, as in \cite{belovs2013quantum,piddock2019quantum}.
For a source vertex $s$ and sink $M$, it is defined as
\[
U
= \swap \cdot \bigoplus_{x \notin \{s,M\}} D_x.
\]
This corresponds to a product of reflections around $\ell_{s,M}^*(E) \coloneqq \linspan\{\ket{\phi_x} \mid x \notin \{s,M\} \}$ and $\ell^+(E)$, based on which we can easily characterize the invariant subspace of $U$.

\begin{lemma} \label{lem:inv-U}
The invariant subspace of $U$ is spanned by the electric $s$-$M$ flow $\ket{f}$ and the flows $\ket{g}$ that are closed outside $M$.
\end{lemma}
\begin{proof}
First note that the product of reflections around subspaces $A = \ell^*_{s,M}(E)$ and $B = \ell^+(E)$ has an invariant subspace spanned by vectors in $(A \cap B) \cup (A^\perp \cap B^\perp)$.

Then, notice that $\ell_{s,M}^*(E) \cap \ell^+(E)$ is empty when $G$ is connected, while $\ell(E)_{s,M}^{*\perp} \cap \ell^-(E)$ contains exactly the flows perpendicular to all star states $\ket{\phi_x}$ for $x \notin \{s,M\}$.
These are the flows that are closed outside of $s$ and $M$, and these are spanned by the electric $s$-$M$ flow $\ket{f}$ and the flows $\ket{g}$ that are closed outside $M$.
\end{proof}

The following is a useful consequence of the lemma above.
Let $\Pi_0 = \ket{f}\bra{f} + \Pi_\cl$ be the projector onto the invariant subspace of $U$, where $\ket{f}$ is the electric $s$-$M$ flow and $\Pi_\cl$ is the projector onto the subspace of closed flows outside $M$.
Let $\ket{\phi_s^-} = \frac{I - \swap}{\sqrt{2}} \ket{\phi_s}$ so that $\ket{\phi_s^-} \in \ell^-(E)$.
Then note that $\Pi_\cl \ket{\phi^-_s} = 0$, and hence
\begin{align*}
\Pi_0 \ket{\phi^-_s}
&= \big( \ket{f}\bra{f} + \Pi_\cl \big) \ket{\phi^-_s} \\
&= \ket{f} \braket{f|\phi^-_s}
= \sqrt{\frac{1}{d_s R_s}} \ket{f},
\end{align*}
where the last equality follows from explicit calculation.

\subsection{Quantum elfs} \label{sec:qwsampling}

From last section it is clear that if we can project the initial state $\ket{\phi^-_s}$ onto the invariant subspace of the quantum walk operator $U$, then we can prepare a quantum state $\ket{f}$ corresponding to the electric flow.
Measuring this state returns an edge $e$ with probability $f_e^2/(w_e R_s)$, i.e., it returns an electric flow sample.

We can approximately implement this projection by doing quantum phase estimation.
This leads to the following algorithm, where the escape time quantity $\ET_s$ pops up very naturally.
The algorithm and its analysis largely follow that of Piddock \cite{piddock2019quantum}, extending it from bipartite graphs to general graphs.
For simplicity, we assume that we know the escape time $\ET_s$.

\begin{algorithm}[H]
\caption{\bf Quantum walk algorithm for elfs} \label{alg:QW-elfs}
\begin{algorithmic}
\vspace{-1mm}
\State
\begin{enumerate}
\item
Run phase estimation on the QW operator $U$ and state $\ket{\phi^-_s}$ to precision $\eps/\sqrt{\ET_s}$.
\item
Measure the phase estimation register.
If the estimate returns ``0'', output the resulting state.
Otherwise, go back to step 1.
\end{enumerate}
\end{algorithmic}
\end{algorithm}

Step 1.~can be implemented with $O(\sqrt{\ET_s}/\eps)$ calls to the (controlled) QW operator.
The number of repetitions and the correctness of the algorithm can be analyzed using the following lemma.
It is proven using a refinement of the effective spectral gap lemma from \cite{lee2011quantum}.
\begin{lemma}[{\cite[Lemma 9]{piddock2019quantum}}] \label{lem:eff-gap}
Define the unitary $U = (2\Pi_B-I)(2\Pi_A-I)$ with projectors $\Pi_A,\Pi_B$.
Let $\ket{\psi} = \sqrt{p} \ket{\psi_0} + \Pi_B \ket{\phi}$ be a normalized quantum state such that $U \ket{\psi_0} = \ket{\psi_0}$ and $\ket{\phi}$ is a (unnormalized) vector satisfying $\Pi_A \ket{\phi} = 0$.
Performing phase estimation on $\ket{\psi}$ with operator $U$ and precision $\delta$, and then measuring the phase register, returns output ``0'' with probability $p' \in [p,p+\delta^2 \|\ket{\phi}\|^2]$, leaving a state $\ket{\psi'}$ such that
\[
\frac{1}{2} \big\| \ket{\psi'}\bra{\psi'} - \ket{\psi_0}\bra{\psi_0} \big\|_1
\leq \frac{1}{\sqrt{p'}} \delta \big\|\ket{\phi}\big\|.
\]
\end{lemma}

We can massage our algorithm into the appropriate form.
Define the projectors $\Pi^- = (I-\swap)/2$ and $\Pi^*_{s,M} = \sum_{x \in \{s,M\}} \ket{\phi_x}\bra{\phi_x}$.
We can rewrite
\[
\ket{f}
= \Pi^- \sqrt{\frac{2}{R_s}} \sum_x v_x \sqrt{d_x} \ket{\phi_x},
\]
with $v_x$ the voltage at vertex $x$ in the unit electric flow from $s$ to $M$.
Now define the (unnormalized) state $\ket{\nu} = \sqrt{\frac{2}{R_s}} \sum_{x \neq s} v_x \sqrt{d_x} \ket{\phi_x}$.
Recalling that $v_y = 0$ for $y \in M$, it holds that
\[
(I - \Pi_{s,M}^*) \ket{\nu} = 0
\quad \text{ and } \quad
\Pi^- \ket{\nu}
= \ket{f} - \Pi^- \sqrt{\frac{2}{R_s}} v_s \sqrt{d_s} \ket{\phi_s}
= \ket{f} - \sqrt{R d_s} \ket{\phi^-_s},
\]
where we used that $v_s = R_s$ and $\ket{\phi_s^-} = \frac{I - \swap}{\sqrt{2}} \ket{\phi_s} = \sqrt{2} \Pi^- \ket{\phi_s}$.

Now let
\begin{equation*}
\begin{gathered}
\Pi_A = I - \Pi_{s,M}^*, \qquad \Pi_B = \Pi^- \\
\ket{\psi} = \ket{\phi^-_s}, \qquad
\ket{\psi_0} = \ket{f}, \qquad
\ket{\phi} = \frac{1}{\sqrt{R_s d_s}} \ket{\nu}, \qquad
p = \frac{1}{R_s d_s}.
\end{gathered}
\end{equation*}
Then we can rewrite the QW operator $U = \swap \cdot \bigoplus_{x \notin \{s,M\}} D_x = (2\Pi_B-I)(2\Pi_A-I)$, and it holds that $\ket{\psi} = \sqrt{p} \ket{\psi_0} + \Pi_B \ket{\phi}$.
Moreover, we can bound
\[
\big\| \ket{\phi} \big\|^2
= \frac{1}{R_s d_s} \frac{2}{R_s} \sum_{x \neq s} v_x^2 d_x
\leq \frac{2}{R_s d_s} \ET_s.
\]
Now we apply \cref{lem:eff-gap}, stating that if we apply phase estimation on $\ket{\psi} = \ket{\phi^-_s}$ with respect to the QW operator $U$ to precision $\delta$, and then measure the phase register, we retrieve ``0'' with probability $p'$ such that
\begin{equation} \label{eq:prob-QPE}
\frac{1}{R_s d_s}
\leq p'
\leq \frac{1}{R_s d_s} (1 + \delta^2 \, 2 \ET_s),
\end{equation}
leaving a state $\ket{f'}$ such that
\begin{equation} \label{eq:TV-distance}
\frac{1}{2} \big\| \ket{f'}\bra{f'} - \ket{f}\bra{f} \big\|_1
\leq \delta \sqrt{2 \ET_s}.
\end{equation}
As a consequence, if we pick $\delta = \eps/\sqrt{2 \ET_s}$ then the phase estimation step succeeds with probability $p' \geq 1/(R_s d_s)$, and upon success it returns a state $\ket{f'}$ that is $\eps$-close to the flow state $\ket{f}$.

The above analysis goes through even if we use an \emph{upper bound} $\hat\ET_s \geq \ET_s$ instead of the exact escape time.
Moreover, similar to \cite{belovs2013quantum,apers2019unified,piddock2019quantum}, we show in \cref{app:QW-sampling} that it is possible to modify the graph by adding an extra edge to $s$ so that $d_s \in \Omega(1/R_s)$.
This leads to the following theorem, where the quantum walk operator is now on the modified graph.

\begin{restatable}[QWs]{theorem}{qwsampling}
\label{thm:qwsampling}
Given an upper bound $\hat\ET_s \geq \ET_s$, there is a quantum walk algorithm based on \cref{alg:QW-elfs} that returns a state $\ket{\tilde f}$ such that $\| \ket{\tilde f} - \ket{f} \|_2 \leq \eps$.
It requires $O\left(\sqrt{\hat\ET_s} \left(\frac{1}{\eps} +\log(R_s d_s)\right) \right)$ calls in expectation to the (controlled) quantum walk operator $U$.
\end{restatable}

\noindent
If $\hat\ET_s \in O(\ET_s)$ then this gives a quadratic improvement compared to the expected number of random walk steps $\ET_s$ in the random walk coupling for electric flow sampling (\cref{lem:coupling}).

\subsection{Quantum elfs process}

We can implement the elfs process by repeatedly using the quantum walk sampling algorithm (and, in fact, this was the initial inspiration).
Consider some initial source vertex $s$ and a fixed sink $M$.
Starting from $s$, prepare the electric $s$-$M$ flow state $\ket{f}$ and measure.
Update the source vertex to the resulting vertex, and repeat, until a sink vertex is found.
If we ignore the error, and assume we know estimates $\hat \ET_s \approx \ET_s$ and $\hat R_s \approx R_s$ for every source $s$ along the elfs path, then this process requires an expected number of quantum walk steps
\[
\E\left[ \sum_{j=0}^{\rho-1} \sqrt{\ET_{Y_j}} \right],
\]
where $\{Y_0 = s,Y_1,\dots,Y_\rho \in M\}$ denotes the elfs process.
We can bound this as follows:
\begin{lemma} \label{lem:sum-sqrt}
$\EHT_s \le \E\left[ \sum_{j=0}^{\rho-1} \sqrt{\ET_{Y_j}} \right]
\le \sqrt{ 2 \, \HT_s \, \EHT_s} \le 2 \, \HT_s.$
\end{lemma}
\begin{proof}
The first inequality follows trivially from the fact that $\ET_x \ge 1$ for all $x$ and $\E[\rho] = \EHT_s$. 
The second inequality is an application of Jensen's inequality (Lemma~\ref{lem:Jensensum} with $f(x)=\sqrt{x}$) and the fact that $\E\left[\sum_{j=0}^{\rho-1} \ET_{Y_j}\right] = 2 \HT_s$ (see \cref{eq:Esumx}).
The third inequality follows from the fact that $\EHT_s \le 2 \HT_s$ (again using \cref{eq:Esumx}).
\end{proof}

The expected number of quantum walk steps is hence upper bounded by twice the random walk hitting time.
If the electric hitting time is small (as is the case for trees), then it is almost quadratically smaller than the hitting time.

Taking into account the errors, we conjecture that there exists a quantum walk algorithm that samples from a distribution $\eps$-close to the elfs process with a runtime of the form $\tO(\sqrt{\HT_s \EHT_s}/\poly(\eps))$.
A particular appeal of this is that it is always upper bounded by the hitting time (see also next section).
Falling short of proving this, we do prove a slightly weaker bound.
\begin{theorem} \label{thm:QW-elfs}
Given upper bounds $\hat\HT_s \geq \HT_s$ and $\hat\EHT_s \geq \EHT_s$, we can implement the elfs process from $s$ to $M$ up to error $\eps$ using
\[
O\left( \frac{\sqrt{\hat\HT_s}}{\eps^2} \hat\EHT_s^{3/2} \log\left( \frac{\hat\HT_s}{\hat\EHT_s} \right) \right)
\]
quantum walk steps.
This yields a sample $\eps$-close to the random walk arrival distribution from $s$ to $M$.
\end{theorem}

First we assume that in each step $Y_j$ of the elfs process we know $R_{Y_j} d_{Y_j}$.
Our algorithm then approximates every step by using the quantum walk algorithm with a fixed budget $T$ per step (i.e., in each step we implement phase estimation with precision $1/T$).
If $Y_1$ is an electric flow sample from $s$, then the QW algorithm on the modified graph (see \cref{app:QW-sampling} and \cref{eq:TV-distance}) gives a sample $Z_1$ satisfying
\begin{equation} \label{eq:TV-QW}
d_{TV}(Y_1,Z_1)
\leq \frac{\sqrt{\ET_s}}{T}.
\end{equation}
Now, for some time parameter $\Gamma$, let $Y = \{Y_0=s,Y_1,\dots,Y_\Gamma\}$ denote the elfs process and let $Z = \{Z_0=s,Z_1,\dots,Z_\Gamma\}$ denote the Markov chain induced by the approximate quantum walk algorithm.
We can prove the following, where $\rho$ is the electric hitting time.
\begin{lemma}
For any $\Gamma > 0$, $d_{TV}(Y,Z) \leq \frac{1}{T} \E\left[ \sum_{j=0}^{\Gamma-1} \sqrt{\ET_{Y_j}} \right] \leq \frac{1}{T} \E\left[ \sum_{j=0}^{\rho-1} \sqrt{\ET_{Y_j}} \right].$
\end{lemma}
\begin{proof}
We consider a coupling between $Y$ and $Z$, and use the bound
\[
d_{TV}(Y,Z)
\leq \Pr(Y \neq Z)
= 1 - \Pr(Y = Z).
\]
We can rewrite
\[
\Pr(Y = Z)
= \Pr(Y_0 = Z_0) \; \Pr(Y_1 = Z_1 \mid Y_0 = Z_0)
    \; \dots \; \Pr(Y_\Gamma = Z_\Gamma \mid Y_{\Gamma-1} = Z_{\Gamma-1}),
\]
where we used that both $Y$ and $Z$ are Markov chains.
Using the shorthand $\eps_{j+1} = 1-\Pr(Y_{j+1} = Z_{j+1} \mid Y_j = Z_j)$ we can bound
\[
\Pr(Y = Z)
= (1-\eps_0) (1-\eps_1) \dots (1-\eps_\Gamma)
\geq 1 - \eps_0 - \eps_1 - \dots - \eps_\Gamma,
\]
where the inequality corresponds to the generalized Bernoulli inequality.
This gives the bound
\[
d_{TV}(Y,Z)
\leq \eps_0 + \eps_1 + \dots + \eps_\Gamma.
\]
Since $Z_0 = Y_0$ we have $\eps_0 = 0$, and by \cref{eq:TV-QW} we have that
\[
\eps_{j+1}
= 1 - \Pr(Y_{j+1} = Z_{j+1} \mid Y_j = Z_j)
\leq \E_{Y_j} \left[\frac{\sqrt{\ET_{Y_j}}}{T}\right],
\]
where $Y_j$ is drawn from the elfs process.
This gives
\begin{align*}
d_{TV}(Y,Z)
&\leq \eps_0 + \eps_1 + \dots + \eps_\Gamma \\
&\leq \sum_{j=0}^{\Gamma-1} \E_{Y_j}\left[ \frac{\sqrt{\ET_{Y_j}}}{T} \right]
= \frac{1}{T} \E_{(Y_0,\dots,Y_\Gamma)}\left[ \sum_{j=0}^{\Gamma-1} \sqrt{\ET_{Y_j}} \right].
\qedhere
\end{align*}
\end{proof}

By \cref{lem:sum-sqrt} we can bound the right hand side by $\sqrt{2 \, \HT_s \, \EHT_s}$.
Hence, if we pick $T \geq \sqrt{2 \, \HT_s \, \EHT_s}/\eps$ we get $d_{TV}(Y,Z) \leq \eps$.
In particular, this implies that the arrival distribution of the QW process will be $\eps$-close to the arrival distribution of the elfs process (and hence the random walk arrival distribution).

To bound the expected number of quantum walk steps, we assume that we have upper bounds $\hat\HT_s \geq \HT_s$ and $\hat\EHT_s \geq \EHT_s$.
We set $T = \sqrt{2 \, \hat\HT_s \, \hat\EHT_s}/\eps$ and run every iteration with $T$ quantum walk steps per iteration.
By the lemma above, this yields a sample $Z_\Gamma$ that is $\eps$-close to $Y_\Gamma$.
By Markov's inequality we have $\Pr(\rho \leq \Gamma) \geq 1 - \E[\rho]/\Gamma$, and so if we set $\Gamma = \hat\EHT_s/\eps$ then $Y_\Gamma$ will be $\eps$-close (and so $Z_\Gamma$ will be $2\eps$-close) to the arrival distribution.
Assuming we have good estimates $R_{Z_j} d_{Z_j}$ of the effective resistances along the path, the total number of quantum walk steps is then
\[
\Gamma \cdot T
\in O\left( \sqrt{\hat\HT_s} \, \hat\EHT_s^{3/2}/\eps^2 \right).
\]

If we do not have estimates of $R_{Z_j} d_{Z_j}$ along the path but we do know the vertex degrees $d_{Z_j}$ of the vertices along the elfs path, then we can estimate the resistances as in \cref{thm:qwsampling}.
This gives an $O(\log(d_{Z_j} R_{Z_j}))$-overhead per iteration, and yields a total expected cost of the order
\[
\E_Z\left[ \sum_{j=0}^{\Gamma-1} T \log(d_{Z_j} R_{Z_j}) \right],
\]
where we use the convention that $\log(d_{Z_j} R_{Z_j}) = 0$ for $Z_j \in M$.
We bound this using the following lemma.
\begin{lemma} \label{lem:Z-sum}
If $d_{TV}(Y,Z) \leq \eps$ then
\[
\Pr_Z\left( \sum_{j=0}^{\Gamma-1} \log(d_{Z_j} R_{Z_j}) > \frac{1}{\eps} \EHT_s \log\left(\frac{\HT_s}{\EHT_s}\right) \right)
\leq 2 \eps.
\]
\end{lemma}
\begin{proof}
First consider the related quantity $\sum_{j=0}^{\Gamma-1} \log(d_{Y_j} R_{Y_j})$ with $Y$ the elfs process.
We can prove that
\begin{equation} \label{eq:Y-sum}
\E_Y\left[ \sum_{j=0}^{\Gamma-1} \log(d_{Y_j} R_{Y_j}) \right]
\le \E_Y\left[ \sum_{j=0}^{\rho-1} \log(d_{Y_j} R_{Y_j}) \right]
\le \EHT_s \log\left(\frac{\HT_s}{\EHT_s}\right)
\end{equation}
where the first inequality follows from the convention that $\log(d_{Y_j} R_{Y_j}) = 0$ for $Y_j \in M$.
For the second inequality we first note that $d_{Y_j} R_{Y_j} \leq \ET_{Y_j}$ (\cref{eq:xbound}), and then use \cref{lem:Jensensum} with $f(x)=\log(x)$ to bound
\[
\E_Y\left[ \sum_{j=0}^{\rho-1} \log(d_{Y_j} R_{Y_j}) \right]
\le \E_Y\left[ \sum_{j=0}^{\rho-1} \log(\ET_{Y_j}) \right]
\le \EHT_s \log\left(\frac{\E_Y\left[ \sum_{j=0}^{\rho-1} \ET_{Y_j} \right]}{\EHT_s}\right).
\]
Finally we use that $\E_Y\left[ \sum_{j=0}^{\rho-1} \ET_{Y_j} \right] = 2\HT_s$ (\cref{eq:Esumx}) to get the claimed bound in \cref{eq:Y-sum}.

By Markov's inequality, \cref{eq:Y-sum} implies that
\[
\Pr_Y\left[ \sum_{j=0}^{\Gamma-1} \log(d_{Y_j} R_{Y_j}) > \frac{1}{\eps} \EHT_s \log\left(\frac{\HT_s}{\EHT_s}\right) \right]
\leq \eps.
\]
Since $d_{TV}(Y,Z) \leq \eps$, this implies that
\[
\Pr_Z\left[ \sum_{j=0}^{\Gamma-1} \log(d_{Z_j} R_{Z_j}) > \frac{1}{\eps} \EHT_s \log\left(\frac{\HT_s}{\EHT_s}\right) \right]
\leq 2\eps. \qedhere
\]
\end{proof}

From this lemma, it follows that we can truncate the cumulative time of the quantum algorithm over the different iterations at
\[
\frac{T}{\eps} \hat\EHT_s \log\left(\frac{\hat\HT_s}{\hat\EHT_s}\right)
= \frac{\sqrt{2 \hat\HT_s}}{\eps^2} \hat\EHT_s^{3/2} \log\left(\frac{\hat\HT_s}{\hat\EHT_s}\right).
\]
If the algorithm has not finished within this time, we output an arbitrary vertex.
By \cref{lem:Z-sum} this truncation only modifies the outcome with probability $2\eps$.
As the original output distribution was $2\eps$-close to the elfs process, the output distribution of the truncated algorithm will be $4\eps$-close to the elfs process.
This proves \cref{thm:QW-elfs}.

\section{Applications}

In this final section we discuss some direct applications of the quantum elfs process.

\subsection{Quantum walk search}

There is a vast literature on what is called \emph{quantum walk search}, which is the use of quantum walks on graphs to find distinguished or ``marked'' vertices \cite{apers2019unified}.
The key idea is that in different scenarios quantum walks can have a quadratically improved hitting time as compared to random walks.
This can be interpreted as a local or graph-theoretic analogue of the quadratic speedup in Grover's quantum search algorithm.
Despite the long line of progress, many key questions remain unanswered.
What can we say about the arrival distribution of a quantum walk?
Can quantum walks always achieve a speedup over random walks?

Our work sheds new light on these questions.
On the one hand, we describe a quantum walk search algorithm that retains the same arrival distribution as a random walk.
This contrasts with former works, where typically nothing could be said about the quantum walk arrival distribution.
We note that the random walk arrival distribution plays a critical role in numerous algorithms, from sampling random spanning trees \cite{broder1989generating} and semi-supervised learning \cite{zhu2003semi} to linear system solving \cite{chung2015solving}.
Our work opens an avenue towards quantum walk speedups for all of these problems.
Moreover, as discussed in \cref{sec:stopping-rules}, if we are given appropriate ``stopping probabilities'' as in the local stopping rule, then our quantum walk algorithm can speed up the access time to \emph{any distribution} $\gamma$.
This complements a recent result on quantum walk speedups over stopping times \cite{bencivenga2020sampling} which assumed access to the target probabilities $\{\gamma_x\}$, rather than the stopping probabilities.\footnote{Both settings are incomparable. E.g., when $\gamma$ is the random walk arrival distribution, then the probabilities $\gamma_x$ are typically unknown but the stopping probabilities are trivial ($p_x = 1$ if $x \in M$ and $p_x = 0$ otherwise).}

On the other hand, we make progress on the question whether quantum walks can always achieve a quantum speedup over the random walk hitting time.
If we have upper bounds $\hat\HT_s$ and $\hat\EHT_s$ on the hitting time and electric hitting time, respectively, then \cref{thm:QW-elfs} shows how to find a marked element $\eps$-close to the arrival distribution in $\tO(\sqrt{\hat\HT_s}\EHT_s^{3/2}/\eps^2)$ quantum walk steps.
We also conjectured that there exists a quantum walk algorithm with runtime $\tO(\sqrt{\hat\HT_s \EHT_s}/\poly(\eps))$.
This is always upper bounded by the hitting time (up to polylog-factors and $\eps$-dependency), and is close to a quadratic speedup when the electric hitting time is small.

A differentiating factor with former quantum walk search algorithms is that our quantum elfs process requires upper bounds on the hitting time and electric hitting time.
This seems unavoidable if we want to sample a marked element from the right arrival distribution.
If, however, we do not care about the arrival distribution, then we do not need upper bounds on $\HT_s$ and $\EHT_s$ as we show in the following theorem.\footnote{Intuitively, this is because we can verify that we found a marked vertex, but not that we sampled it from the right distribution.}
\begin{theorem}
Given a source vertex $s$ on a graph, and query access\footnote{We make queries of the form ``Is $x \in M$''? No other prior knowledge of $M$ is required.} to the vertices in the sink $M$.
There is a quantum walk algorithm that returns an element from the sink $M$ using a number of quantum walk steps
\[
\tO(\sqrt{\HT_s} \EHT_s^{3/2}).
\]
\end{theorem}
\begin{proof}
This essentially follows from invoking \cref{thm:QW-elfs} with a fixed $\eps = 1/10$ and geometrically varying guesses.

Let the function $T(\tilde\HT_s,\tilde\EHT_s) = \tO(\sqrt{\tilde\HT_s} \tilde\EHT_s^{3/2})$ denote the runtime of the algorithm in \cref{thm:QW-elfs} with precision $\eps = 1/10$ and guesses $\tilde\HT_s$ and $\tilde\EHT_s$.
We know that if $\tilde\HT_s \geq \HT_s$ and $\tilde\EHT_s \geq \EHT_s$ then the algorithm returns a marked element with probability at least $9/10$.

To vary over parameter guesses, we consider an outer loop over geometrically increasing budgets $b_i = 2^i$, for $i = 0,1,2,\dots$.
For a fixed budget $b_i$, we run an inner loop over possible upper bound guesses $\tilde\HT_s$ and $\tilde\EHT_s$ compatible with that budget.
Specifically, for $0 \le j \le i$ we set $\tilde\HT_s = b_i/2^j$ ($0 \le j \le i$), and we pick $\tilde\EHT_s$ such that $f(\tilde\HT_s,\tilde\EHT_s) \in [b_i,2b_i)$.
For every budget $b_i$, there are at most $i+1$ guesses for $\tilde\HT_s$ and $\tilde\EHT_s$, so that the inner loop has cost $O((i+1) b_i)$.

By definition, if the budget $b_i \geq f(\HT_s,\EHT_s)$ then the inner loop must encounter a pair of guesses $\tilde\HT_s \geq \HT_s$ and $\tilde\EHT_s \geq \EHT_s$, in which case the algorithm returns a marked element with probability at least $9/10$ (by our choice of $\eps = 1/10$).
We can now bound the total expected complexity.
Let $i^* \in \tO(\log(\sqrt{\HT_s} \EHT_s^{3/2}))$ be such that the budget $b_{i^*} = 2^{i^*}$ is in $[f(\HT_s,\EHT_s),2f(\HT_s,\EHT_s))$.
Then the total expected cost is bounded by
\begin{align*}
\sum_{i=0}^{i^*-1} (i+1) 2^i
&+ \sum_{i\ge i^*} (i+1) 2^i \frac{1}{10^{i-i^*}} \\
&= \sum_{i=0}^{i^*-1} (i+1) 2^i + 2^{i^*} \sum_{i\ge i^*} (i+1) \frac{1}{5^{i-i^*}} \\
&= \sum_{i=0}^{i^*-1} (i+1) 2^i + 2^{i^*} \sum_{i\ge i^*} (i-i^*+1) \frac{1}{5^{i-i^*}}
+ i^* 2^{i^*} \sum_{i\ge i^*} \frac{1}{5^{i-i^*}}.
\end{align*}
We can bound these terms using that for $x \neq 1$
\[
\sum_{\ell=0}^{k-1} (\ell+1) x^\ell
= \frac{\mathrm{d}}{\mathrm{d} x} \left(\sum_{\ell=0}^k x^\ell \right)
= \frac{\mathrm{d}}{\mathrm{d} x} \frac{1 - x^{k+1}}{1-x}
= \frac{k x^{k+1} - (k+1) x^k + 1}{(1-x)^2}.
\]
This shows that $\sum_{i=0}^{i^*-1} (i+1) 2^i = i^* 2^{i^*} - 2^{i^*} + 1 \leq i^* 2^{i^*}$, $2^{i^*} \sum_{i\ge i^*} (i-i^*+1) \frac{1}{5^{i-i^*}}
= \frac{2^{i^*}}{(1-1/5)^2}$, and finally by the usual geometric series we have that $i^* 2^{i^*} \sum_{i\ge i^*} \frac{1}{5^{i-i^*}} = \frac{i^* 2^{i^*}}{1-1/5}$.
The overall expected complexity is hence $O(i^* 2^{i^*}) \in \tO(\sqrt{\HT_s} \EHT_s^{3/2})$.
\end{proof}

For the special case of trees this yields a near-quadratic speedup over the random walk hitting time.
Indeed, in \cref{part:trees} we showed that the electric hitting time on trees is logarithmic in the graph size and edge weights.
Combining this with \cref{thm:QW-elfs} yields the following corollary.
\begin{corollary}
Let $G = (V,E,w)$ be a tree graph, and $M \subseteq V$ a subset of marked or sink vertices.
From an initial vertex $s$, there is a quantum algorithm that makes $\tO(\sqrt{\HT_s})$ quantum walk steps\footnote{Here the $\tO(\cdot)$-notation also hides polylogarithmic factors in the tree size $n$ and maximum ratio $w_{\max}/w_{\min}$ of the edge weights.} and returns an element from $M$ with constant probability.
\end{corollary}

\subsection{Approximating effective resistances}

As a second application we describe a quantum walk algorithm for estimating the escape probability or effective resistance $R_s$.
The algorithm gives a quadratic speedup over a similar random walk algorithm based on the escape time from \cref{sec:RW-eff-res}.

The algorithm is a variation on the quantum walk algorithms described in \cite{ito2019approximate,piddock2019quantum}, which have a complexity scaling as $\tO(\sqrt{\CT_s}/\eps^{3/2})$.
These algorithms are most naturally compared to the random walk algorithm based on the \emph{hitting time} from \cref{sec:RW-eff-res}, which has a runtime $O(\HT_s/\eps^2)$.
While $\CT_s \geq \HT_s$, these quantities are in general incomparable, and the quantum algorithm only has a speedup when $\CT_s \in o(\HT_s^2)$.

We propose a quantum walk algorithm with complexity $\tO(\sqrt{\hat\ET_s}/\eps^{3/2})$ for a given upper bound $\hat\ET_s \geq \ET_s$.
This gives a quadratic speedup over the corresponding random walk algorithm from \cref{sec:RW-eff-res}, which has complexity $\tO(\hat\ET_s/\eps^3)$ under the same assumptions.
It also improves over the prior quantum walk algorithms if we have access to a bound $\ET_s \leq \hat\ET_s \leq \CT_s$.

\begin{theorem}
Given an upper bound $\hat\ET_s \geq \ET_s$, there is a quantum walk algorithm that $\eps$-multiplicatively estimates $R_s d_s$ using a number of quantum walk steps
\[
O\left(\sqrt{\hat\ET_s} \left( \frac{1}{\eps^{3/2}} + \log(R_s d_s) \right) \right).
\]
\end{theorem}
\noindent
As mentioned in \cref{sec:RW-eff-res}, since the escape probability $p_s = 1/(R_s d_s)$, this gives a multiplicative estimate of $p_s$.
If we are also given $d_s$ (or an $\eps$-approximation of $d_s$), then we can derive an $\eps$-approximation of the effective resistance $R_s$.

The quantum walk algorithm starts from the initial state $\ket{\phi_s^-}$.
Let $Q_\delta$ denote the phase estimation operator with respect to QW operator $U$ and precision $\delta$.
From \cref{eq:prob-QPE} we know that $Q_\delta \ket{\phi_s^-}$ has probability $p' = (1 \pm \delta^2 \ET_s)/(R_s d_s)$ of having ``0'' in its phase estimation register.
Setting $\delta = \sqrt{\eps/\hat\ET_s}$ for any upper bound $\hat\ET_s \geq \ET_s$, this is $p' = (1 \pm \eps)/(R_s d_s)$.
We can then apply quantum amplitude estimation on $\ket{\phi}$ and $U$ \cite[Theorem 12]{brassard2002quantum} to get an $\eps$-multiplicative estimate of $p'$, at the cost of $O(1/(p' \eps))$ calls to $U$.
The cost of $U$ corresponds to $O(1/\delta)$ quantum walk steps, and so the overall complexity is a number of quantum walk steps that scales with
\[
\frac{1}{p' \eps} \frac{1}{\delta}
\in O\left( \frac{\sqrt{\hat\ET_s}}{\eps^{3/2}} R_s d_s \right).
\]
Using the same idea as in \cref{app:QW-sampling} we can improve this complexity to the claimed $O\big(\sqrt{\hat\ET_s} \big( \frac{1}{\eps^{3/2}} + \log(R_s d_s) \big) \big)$ by modifying the graph.

\bibliographystyle{alpha}
\bibliography{bibliography.bib}

\appendix
\addtocontents{toc}{\protect\setcounter{tocdepth}{0}}

\appendix
\addtocontents{toc}{\protect\setcounter{tocdepth}{0}}

\section{Alternative proof of escape time identity} \label{app:escape-time-doob}

Let $P'$ denote the Markov chain that is absorbing in $\{s,M\}$.
The function $h(x) = \Pr_x\{\tau_s < \tau_M\} = v_x/R_s$ is harmonic with respect to $P'$ on $U \backslash \{s\}$.
As a consequence, we can define the \emph{Doob $h$-transform}\footnote{For more details on the Doob transform, see e.g.~\cite[Section 17.6]{levin2017markov}.} $\hat P$ of $P$ with respect to $h$ by
\[
\hat P(x,y)
= \frac{P(x,y) h(y)}{h(x)}
= \frac{P(x,y) v_y}{v_x}, \qquad x \in U.
\]
It is easy to check that this describes a stochastic transition matrix on $U$ that is absorbing in $s$.
Using that $v_x = R_s \Pr_x[\tau_s < \tau_M]$, we have that for $x \in U$,
\[
\hat P(x,y)
= \frac{P(x,y) \Pr_y[\tau_s < \tau_M]}{\Pr_x[\tau_s < \tau_M]}
= \frac{\Pr_x[X_1=y,\tau_s < \tau_M]}{\Pr_x[\tau_s < \tau_M]}
= \Pr_x[X_1=y \mid \tau_s < \tau_M].
\]
As a consequence, the transformed Markov chain with transition matrix $\hat P$ corresponds to the original chain \emph{conditioned} on hitting $s$ before $M$.

We can alternatively describe $\hat P$ as a random walk that is absorbing in $s$ on a graph $\hat G = (U,\hat E,\hat w)$ with weights
\[
\hat w_{x,y}
= w_{x,y} h(x) h(y)
= w_{x,y} \frac{v_x v_y}{R^2_{s,M}}.
\]
For the \emph{non-absorbing} walk on $\hat G$, we can calculate the degree $\hat d_x$ for $x \neq s$ as
\[
\hat d_x
= \sum_y \hat w_{x,y}
= h(x) \sum_y w_{x,y} h(y)
= d_x h(x)^2
= \frac{d_x}{R_s^2} v_x^2,
\]
where we used the harmonic condition $h(x) = \sum_y \frac{w_{x,y}}{d_x} h(y)$.
At $s$, $h$ is not harmonic, but since $v_y$ are the voltages for a unit flow from $s$, we have $\sum_y w_{sy} (v_s-v_y)=1$ and so
\[\hat d_s = \sum_y \hat w_{s,y} = \sum_y w_{s,y} \frac{v_y}{R_s}
= \frac{1}{R_s}\left(\sum_y w_{sy} v_s -1\right)
= d_s- \frac{1}{R_s}.\]
We can calculate the total edge weight $\hat W$ as
\begin{equation} \label{eq:Doob-weight}
2 \hat W
= \sum_x \hat d_x
= \frac{1}{R_s^2} \sum_{x} d_x v_x^2 - \frac{1}{R_s}=\frac{x_s-1}{R_s}.
\end{equation}

We can now easily prove the escape time identity in \cref{lem:escape-time}, which states that the expected value of the escape time $\esc = 1 + \max\{t \mid X_t = s\}$ on $G$ equals the quantity $\ET_s = \frac{1}{R_s} \sum_x v_x^2 d_x$.

\begin{proof}[Proof of \cref{lem:escape-time} (escape time identity)]
We first find an expression for the return time $\tau^+_s$ of the original walk from $s$, \emph{conditioned} on returning to $s$ before hitting $M$ (i.e., $\tau^+_s < \tau_M$).
This is equal to the expected return time $\hat\tau^+_s$ of the Doob transformed chain $\hat P$.
By Kac's lemma (\cref{eq:exp-return}) and \cref{eq:Doob-weight} this is
\[
\E_s(\tau_s^+ | \tau_s^+ < \tau_M)
= \E_s(\hat\tau_s^+)
= \frac{1}{\hat \pi(s)}
= \frac{2 \hat W}{\hat d_s}
= \frac{x_s-1}{ R_sd_s -1}.
\]
Now recall that $\Pr_s[\tau_s^+ > \tau_M]=1/R_sd_s$ (\cref{eq:return-prob}), and hence 
\begin{align*}
\E_s(\esc)
&= \E_s(1+\max\{ t \mid X_t = s\}) \\
&= 1 + \Pr_s[\tau_s^+ > \tau_M] \times 0+\sum_{j=1}^{\infty}(1-\Pr_s[\tau_s^+ > \tau_M])^j\E_s(\tau_s^+ | \tau_s^+ < \tau_M) \\
&= 1 + \frac{1-\Pr_s[\tau_s^+ > \tau_M]}{\Pr_s[\tau_s^+ > \tau_M]}\E_s(\tau_s^+ | \tau_s^+ < \tau_M)
= x_s.
\end{align*}
This proves \cref{lem:escape-time}.
\end{proof}

\section{Edge coupling lemma} \label{app:edge-coupling}

Here we prove the edge coupling lemma, which we recall below.
\edgecoupling*
\begin{proof}
Similar to the (vertex) coupling rule (\cref{lem:coupling}), we can equivalently describe the edge stopping rule $\mu$ by a random walk on a modified graph $\hat G$.
The graph is obtained by (i) splitting every edge $(x,y)$ of $G$ into a pair of edges $(x,(xy))$ and $((xy),y)$ with weights $w_{x,(xy)} = w_{(xy),y} = w_{xy}$, and then (ii) adding another vertex $(xy)'$ that is connected to the new vertex $(xy)$ with an edge of weight
\[
w_{(xy),(xy)'}
= \frac{(v_x-v_y)^2}{v_x v_y} w_{xy}.
\]
Note that for $y \in M$ we get $w_{(xy),(xy)'} = \infty$.
We let $\hat M = \{ (xy)' \mid (x,y) \in E \}$ denote the new sink\footnote{While the random walk on $\hat G$ is also absorbing in $M$, it will never hit a vertex in $M$ as $w_{(xy),(xy)'} = \infty$ for $y \in M$.}.
A random walk from $x$ in $\hat G$ will pick a random outgoing edge $(x,(xy))$ with probability $w_{xy}/d_x$, after which it will jump to the sink vertex $(xy)'$ with probability
\[
\frac{w_{(xy),(xy)'}}{2 w_{xy} + w_{(xy),(xy)'}}
= \frac{(v_x-v_y)^2}{v_x^2+v_y^2},
\]
and otherwise jump to either $x$ or $y$, each with probability $1/2$.
As a consequence, a random walk on $\hat G$ from vertex $s$ to $\hat M$ will stop at vertex $(xy)'$ corresponding to edge $(x,y)$ with probability $\Pr_s[Z_\mu = (x,y)]$.

It remains to analyze the arrival probability in the sink $\hat M$.
For this, we claim that the unit electric flow from $s$ to $\hat M$ is induced by the vertex potentials $\hat v_x = v_x^2/R_s$, $\hat v_{(xy)} = v_x v_y/R_s$ and $\hat v_{(xy)'} = 0$.
Indeed, for any original vertex $x$ we have an outgoing flow of 
\[
\sum_y \hat f_{x,(xy)}
= \sum_y (\hat v_x - \hat v_{(xy)}) w_{xy}
= \frac{v_x}{R_s} \sum_y (v_x - v_y) w_{xy}
= \delta_{sx}.
\]
For a new vertex $(xy)$ we have an outgoing flow of
\begin{align*}
\hat f_{(xy),x} + \hat f_{(xy),y} + \hat f_{(xy),(xy)'}
&= (\hat v_{(xy)} - \hat v_x) w_{xy} + (\hat v_{(xy)} - \hat v_y) w_{xy} + \hat v_{(xy)} w_{(xy),(xy)'} \\
&= \frac{w_{xy}}{R_s} \left( (v_xv_y - v_x^2) + (v_xv_y - v_y^2) + v_x v_y \frac{(v_x-v_y)^2}{v_xv_y} \right)
= 0.
\end{align*}
This proves that indeed these potentials induce a unit electric flow from $s$ to $M$.
We can now calculate the incoming flow in a vertex $(xy)'$ as
\[
\hat f_{(xy),(xy)'}
= \hat v_{(xy)} w_{(xy),(xy)'}
= \frac{1}{R_s} (v_x-v_y)^2 w_{xy},
\]
which equals the probability of sampling edge $(x,y)$ from the electric flow from $s$ to $M$ in~$G$.
As this equals the arrival probability of a random walk from $s$ to $\hat M$ on $\hat G$, this proves the lemma.
\end{proof}

\section{Improved QW sampling algorithm} \label{app:QW-sampling}

Here we prove \cref{thm:qwsampling}, which we restate below.

\qwsampling*

In \cref{sec:qwsampling} of the main text we already described a quantum walk algorithm with complexity $O(R_s d_s \sqrt{\hat\ET_s}/\eps)$, given an upper bound $\hat\ET_s \geq \ET_s$.
In the following, we show how to modify the graph so that $R_s d_s \in \Theta(1)$.
First we consider the case where we have an estimate of $R_s d_s$.

We will run the QW sampling algorithm on a modified graph $G'$, which equals $G$ except that we append a new vertex $s'$ to $s$ with weight $w_{s s'} = \eta d_s $ (see \cref{fig:R-graph}).

\begin{figure}[htb]
\centering
\includegraphics[width=.4\textwidth]{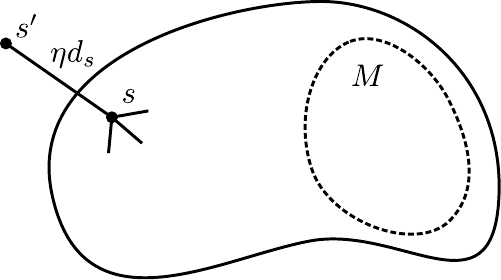}
\caption{Modified graph $G'$.}
\label{fig:R-graph}
\end{figure}

A random walk can be easily implemented on $G'$ using the random walk on $G$ and knowledge of $\eta$ (knowledge of $d_s$ is not required). At vertex $s$, step to $s'$ with probability $\frac{\eta}{1+\eta}$, otherwise take a step of the random walk on $G$. Similarly the quantum walk operator on $G'$ must reflect around $\ket{\phi'_s}$ instead of $\ket{\phi_s}$ where 
\[\ket{\phi'_s}=\sqrt{\frac{\eta}{1+\eta}}\ket{ss'}+\sqrt{\frac{1}{1+\eta}} \ket{\phi_s}\]
We do not consider the precise details of how to implement this operator, but note that if the quantum walk operator is implemented using a unitary $W:\ket{x}\ket{0}\rightarrow \ket{\phi_x}$, we can use a linear combination of unitaries to implement $W':\ket{x}\ket{0}\rightarrow \ket{\phi'_x}$ at a similar cost.

On $G'$ we have that $R_{s'} = R_s+\frac{1}{\eta d_s}$, $d_{s'} = \eta d_s$ and hence $R_{s'} d_{s'}=1+\eta R_s d_s $.
The degree at $s$ in $G'$ is $d'_s=d_s+d_{s'}$, and the degree of all other vertices in $G$ is unchanged.
Moreover, the unit electric flow from $s'$ to $M$ has the original voltages $v_x$ for $x \neq s'$ and $v_{s'} = R_{s'}$, so that the escape time from the new vertex
\begin{align*}
\ET_{s'}
&= \frac{1}{R_{s'}} \left( d_{s'} v_{s'}^2 + d_{s'}v_s^2+\sum_{x \in V} d_x v^2_x \right) \\
&=R_{s'}d_{s'}+d_{s'}R_s\frac{1}{R_{s'}}+\frac{R_{s}}{R_{s'}}\ET_s \\
&=1+\eta R_s d_s+\frac{(\eta R_s d_s)^2}{1+\eta R_s d_s}+\frac{\eta R_s d_s}{1+\eta R_s d_s}\ET_s \\
&\leq 1+2\eta R_s d_s +\ET_s.
\end{align*}
By choosing $\eta \leq 1$, we ensure that $\ET_{s'}\leq 4\ET_s$ (recalling $1\leq R_sd_s \leq \ET_s$ from \cref{eq:xbound}).
\newcommand{\tpesc}{\tilde{p}_s}
\paragraph{Given estimate $\tpesc$.}
Assume that we are given an estimate $\tpesc$ such that $\frac{1}{2 R_s d_s} \leq \tpesc \leq \frac{1}{R_s d_s}$ and an upper bound $\hat\ET_s \geq \ET_s$.
We set $\eta=\tpesc$, so that $R_{s'}d_{s'} =1+\tpesc R_s d_s \leq 2$.

Running the quantum walk sampling algorithm on the modified graph has complexity $O(R_{s'} d_{s'} \sqrt{\hat\ET_s}/\eps^2) = O(\sqrt{\hat\ET_s}/\eps^2)$.
Note that the algorithm prepares an $\eps$-approximation of the flow state $\ket{f'}$ from $s'$ to $M$ instead of the original flow state $\ket{f}$.
However, we can decompose
\[
\ket{f'}
= \frac{1}{\sqrt{2 R_{s'}}} \left( \frac{1}{\sqrt{\eta d_s}} (\ket{ss'}-\ket{s's}) + \sqrt{2R_s} \ket{f} \right).
\]
Since $\braket{ss'|f} = \braket{s's|f} = 0$, we can simply measure the state $\ket{f'}$ and postselect on receiving an edge $(x,y) \neq (s,s'),(s',s)$.
This succeeds with probability $|\braket{f'|f}|^2 = R_s/R_{s'} \geq 1/3$, in which case it returns a sample from $\ket{f}$.

\paragraph{Without a good estimate of $\tpesc$.}
Now assume that we only have an estimate $\tpesc$ satisfying $\frac{1}{2R_sd_s} \leq \tpesc \leq 1$ (e.g., we could trivially set $\tilde p_s = 1$).
We again define the modified graph $G'$ with $\eta = \tpesc$, so that we have $R_{s'} d_{s'} = 1 + \tpesc R_s d_s$.

We can again run the quantum walk algorithm with upper bound $\hat\ET_{s'} = 4 \hat\ET_s \geq \ET_{s'}$.
The same analysis goes through.
We run quantum phase estimation for time $O(\sqrt{\hat\ET_{s'}}/\eps) = O(\sqrt{\hat\ET_s}/\eps)$, and measure the phase estimation register.
By \cref{eq:prob-QPE} this returns ``0'' with probability $p = (1\pm \eps^2)/(R_{s'} d_{s'}) = (1\pm \eps^2)/(1+\tpesc R_s d_s) $.
In that case, we succesfully prepared an $\eps$-approximation of the flow state
\[
\ket{f'}
= \frac{1}{\sqrt{2 R_{s'}}} \left( \frac{1}{\sqrt{\eta d_s}} (\ket{ss'}-\ket{s's}) + \sqrt{2R_s} \ket{f} \right).
\]
The state $\ket{f'}$ now has overlap $|\braket{f'|f}|^2 = R_s/R_{s'} \geq 1/3$.

While $\tpesc \gg 1/(R_s d_s)$, the algorithm has a low success probability.
We remedy this by halving our estimate $\tpesc$ until the success probability becomes constant:

\begin{algorithm}[H]
\caption{\bf Quantum walk algorithm with bounds} \label{alg:QW-upper-bounds}
\begin{algorithmic}
\State
Given: upper bound $\hat \ET_s \geq \ET_s$.
\begin{enumerate}
\item
Set $\tpesc=1$. Construct the modified graph $G'$ with $\eta = \tpesc$ and set $\hat\ET_{s'} = 4 \hat\ET_s$.
\item
Run phase estimation on the QW operator $U$ and state $\ket{\phi^-_{s'}}$ on $G'$ to precision $1/(2\sqrt{\hat\ET_{s'}})$.
\item
Run quantum amplitude estimation to estimate the success probability $p = (1\pm \frac{1}{4})/(1+\tpesc R_s d_s)$ to precision $1/8$.
If the estimate is less than $1/8$, halve $\tpesc$ and go back to step 1.
\item
Run phase estimation on the QW operator $U$ and state $\ket{\phi^-_{s'}}$ on $G'$ to precision $\eps/(2\sqrt{\hat\ET_{s'}})$.
Measure the phase estimation register.
If the estimate returns ``0'', output the resulting state.
Otherwise, go back to step 1.
\end{enumerate}
\end{algorithmic}
\end{algorithm}

We divide the algorithm in a first phase where $\tpesc > 3/(R_sd_s)$, and a second phase where $\tpesc \leq 3/(R_sd_s)$.
During the first phase we have $p \leq (1+\frac{1}{4})/4$, and so from step 3.~we can still return to step 1.
As we halve $\tpesc$ with every repetition, the total number of repetitions of steps 1.-3. is $O(\log(R_s d_s))$.

After that, we have $3/(2R_sd_s) < \tpesc \leq3/(R_sd_s)$ so that $p \geq 1/4$ and our estimate in step 3.~will be at least $1/8$.
At this point, the algorithm will no longer change our estimate $\tpesc$.
In every repetition the algorithm proceeds to step 4., in which it succeeds with probability at least $1/4$.
The expected number of repetitions in this phase is hence constant.

Overall, in expectation the algorithm terminates after $O(\log(R_s d_s))$ iterations of steps 1.-3., each of which has complexity $O(\sqrt{\hat\ET_s})$, and  step 4. has complexity $O(\sqrt{\hat\ET_s}/\eps)$, so the total complexity in expectation is
\[
O\left(\sqrt{\hat\ET_s} \left(\frac{1}{\eps} +\log(R_s d_s)\right) \right).
\]
This proves \cref{thm:qwsampling}.

\section{Jensen's inequality}

We will use the following useful variation of Jensen's inequality.

\begin{lemma} \label{lem:Jensensum}
Consider the elfs process $\{Y_0 = s,Y_1,\dots,Y_\rho \in M\}$ with $\E[\rho] = \EHT_s$.
Let $f$ be a concave function, and let $g \in \mathbb{R}^V$ associate values to all vertices. Then
\[
\E\left[\sum_{i=0}^{\rho-1} f\left(g_{Y_i} \right)\right]
\le	\EHT_s \, f \left(\frac{\E\left[\sum_{i=0}^{\rho-1} g_{Y_i}\right] }{\EHT_s} \right).
\]
\end{lemma}
\begin{proof}
We will use Jensen's inequality, which states that for a concave function $f$, nonnegative weights $\{c_x\}$ and values $\{g_x\}$ we have that
\[
f\left( \frac{\sum_x c_x g_x}{\sum_x c_x} \right)
\geq \frac{\sum_x c_x f(g_x)}{\sum_x c_x}.
\]
Let $c_x \geq 0$ denote the expected time spent by the elfs process in a vertex $x$, so that $\sum_x c_x = \EHT_s$ and
\[
\E\left[\sum_{i=0}^{\rho-1} f\left(g_{Y_i} \right)\right]
= \sum_x c_x f(g_x).
\]
Jensen's inequality then states that
\[
\sum_x c_x f(g_x)
\leq \left(\sum_x c_x\right) f\left( \frac{\sum_x c_x g_x}{\sum_x c_x} \right)
= \EHT_s f\left( \frac{\E\left[\sum_{i=0}^{\rho-1} g_{Y_i} \right]}{\EHT_s} \right). \qedhere
\]
\end{proof}

\end{document}